\documentclass{vldb}
\usepackage{graphicx}
\usepackage{float}
\usepackage[sort,nocompress]{cite}
\usepackage[
font=small
]{subfig}
\usepackage{url}
\usepackage{bbm}
\usepackage{courier}
\usepackage[labelfont=bf]{caption}
\usepackage{balance}  % for  \balance command ON LAST PAGE  (only there!)
\usepackage{enumitem,kantlipsum}
\usepackage{amsmath}
\usepackage{algorithm}
\usepackage[noend]{algpseudocode}
\usepackage{xfrac}
\usepackage{xcolor}
\usepackage{flushend}
\usepackage{mathtools}
\usepackage{multirow}
\makeatletter
\def\BState{\State\hskip-\ALG@thistlm}
\makeatother

\vldbTitle{VIP Hashing - Adapting to Skew in Popularity of Data on the Fly}
\vldbAuthors{Aarati Kakaraparthy, Jignesh M. Patel, Brian P. Kroth, and Kwanghyun Park}
\vldbDOI{https://doi.org/10.14778/xxxxxxx.xxxxxxx}
\vldbVolume{12}
\vldbNumber{xxx}
\vldbYear{2021}

\newcommand{\shepherd}[1]{\textcolor{black}{#1}}

\newcommand{\revision}[1]{\textcolor{black}{#1}}

\renewcommand{\footnotesize}{\small}

\makeatletter
\algrenewcommand\ALG@beginalgorithmic{\small}
\makeatother

\begin{document}
	
	\newlength{\textfloatsepsave}
	\setlength{\textfloatsepsave}{\textfloatsep}
	
	\newtheorem{lemma}{Lemma}
	\newtheorem{theorem}{Theorem}

	% ****************** TITLE ****************************************
	
	\title{VIP Hashing - Adapting to Skew in Popularity of Data on the Fly (extended version)}

	% possible, but not really needed or used for PVLDB:
	%\subtitle{[Extended Abstract]
	%\titlenote{A full version of this paper is available as\textit{Author's Guide to Preparing ACM SIG Proceedings Using \LaTeX$2_\epsilon$\ and BibTeX} at \texttt{www.acm.org/eaddress.htm}}}
	
	% ****************** AUTHORS **************************************
	
	
	\numberofauthors{4} %  in this sample file, there are a *total*
	
\author{
	% 1st. author
	\alignauthor
	Aarati Kakaraparthy\ \ \ \ \ Jignesh M. Patel\\
	\vspace{3pt}
	\affaddr{University of Wisconsin, Madison}\\
	\email{\{aaratik, jignesh\}@cs.wisc.edu}
%	\and
	\alignauthor Brian P. Kroth\ \ \ \ \ Kwanghyun Park\\
	\vspace{3pt}
	\affaddr{Microsoft Gray Systems Lab}\\
	\email{\{bpkroth, kwpark\}@microsoft.com}
}

%\author{
%	\alignauthor
%	Aarati Kakaraparthy\ \ \ \ Jignesh M. Patel\\
%	\vspace{2pt}
%	\affaddr{University of Wisconsin, Madison}\\
%	%   \affaddr{1932 Wallamaloo Lane}\\
%	%   \affaddr{Wallamaloo, New Zealand}\\
%	\email{\{aaratik, jignesh\}@cs.wisc.edu}
%	\alignauthor Kwanghyun Park\ \ \ \ Brian P. Kroth\\\vspace{2pt}
%	\affaddr{Microsoft Gray Systems Lab, Madison}\\
%	\email{\{kwpark, bpkroth\}@microsoft.com}
%}
%	\date{15 January, 2019}

	\maketitle
	
	\begin{abstract}
	All data is not equally popular. Often, some portion of data is more frequently accessed than the rest, which causes a skew in popularity of the data items. Adapting to this skew can improve performance, and this topic has been studied extensively in the past for disk-based settings. In this work, we consider an in-memory data structure, namely \textit{hash table}, and show how one can leverage the skew in popularity for higher performance.

Hashing is a low-latency operation, sensitive to the effects of caching, branch prediction, and code complexity among other factors. These factors make learning in-the-loop especially challenging as the overhead of performing any additional operations can be significant. In this paper, we propose VIP hashing, a \textit{fully online} hash table method, that uses lightweight mechanisms for \textit{learning} the skew in popularity and \textit{adapting} the hash table layout. These mechanisms are non-blocking, and their overhead is controlled by \textit{sensing} changes in the popularity distribution to \textit{dynamically switch-on/off} the learning mechanism as needed.

We tested VIP hashing against a variety of workloads generated by \textit{Wiscer}, a homegrown hashing measurement tool, and find that it improves performance in the presence of skew (\revision{22\% increase} in fetch operation throughput for a hash table with one million keys under low skew, \shepherd{77\% increase under medium skew}) while being robust to insert and delete operations, and changing popularity distribution of keys. \shepherd{We find that VIP hashing reduces the end-to-end execution time of TPC-H query 9, which is the most expensive TPC-H query, by 20\% under medium skew.}	
	\end{abstract}
	
	\section{Introduction}\label{sec:introduction}
	\begin{figure}
	\centering
	\subfloat[Default configuration: VIPs at random spots]{
		\includegraphics[scale=0.16]{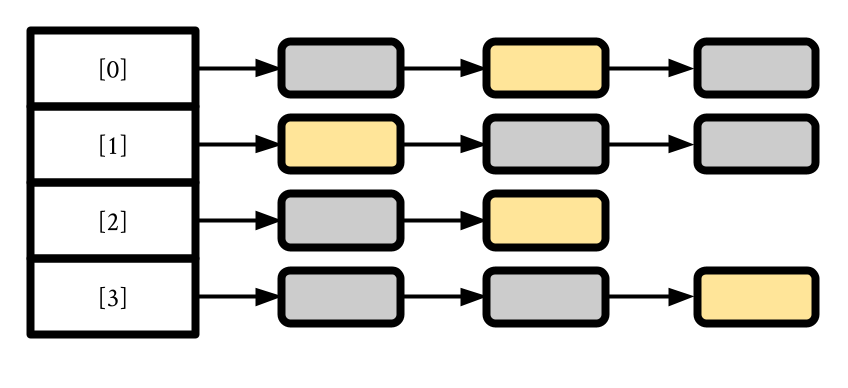}
		\label{fig:random_arrangement}
	}
	\hspace{4mm}
	\subfloat[VIP configuration: VIPs at the front]{
		\includegraphics[scale=0.16]{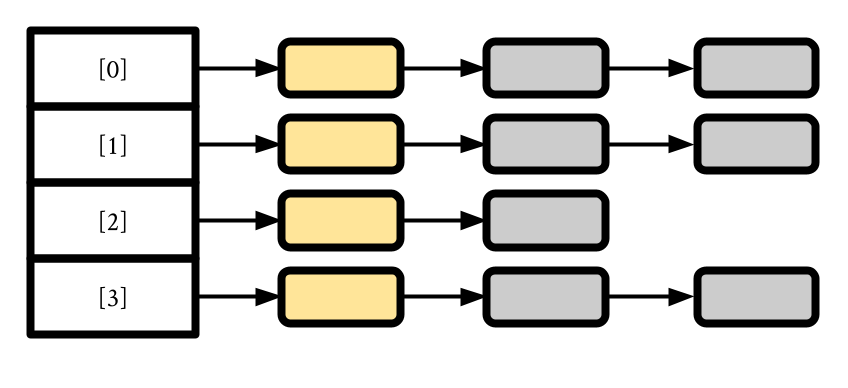}
		\label{fig:vip_arrangement}
	}
	\caption{\textbf{Hash Table configurations with VIP keys (in yellow) at (a) random spots, vs. (b) at the front. The throughput of the hash table can be improved by giving VIPs more favorable spots in the front.}}
	\label{fig:vip_idea}
	\vspace{-2mm}
\end{figure}

Hash tables are widely used data structures with a simple point lookup interface \textendash\ mapping a key to a value. In database systems, they are used for in-memory indexing and also in query processing operations such as hash joins and aggregation. The lightweight computation involved and the constant time lookup guarantees are two reasons that enable hash tables to achieve high throughput when processing point queries.

However, not all keys contribute equally to the performance, and requests are often skewed towards a smaller set of ``hot" keys. In multiple studies involving production workloads, fetch requests have been observed to follow \revision{the power law}~\cite{powerlaw,fbstudy,zipf2} where the popularity of keys exponentially decays with the rank. The Very Important key-value Pairs (VIPs) are the keys with lower rank, as they constitute a larger portion of requests and have a greater impact on the throughput. It is possible to further improve the throughput obtained from the hash table by leveraging this skew in popularity, as we show in our work.

Fig.~\ref{fig:vip_idea} shows the core motivation behind VIP Hashing \textendash\ giving more favorable spots to more popular keys. In the VIP configuration (Fig.~\ref{fig:vip_arrangement}), the keys are ordered in descending order of popularity and the VIPs are in the front, analogous to seating VIPs in the front row for an event. By placing the popular keys at the start, they can be accessed faster due to multiple reasons such as fewer memory accesses and lesser computation (discussed in \S\ref{sec:roofline}), which improves the overall throughput obtained from the hash table.

While attaining the VIP configuration is straightforward if the popularity of keys is known in advance (keys can be inserted in the right position in the chain according to their popularity), one might not have this information up front. Also, the popularity of the keys can change over time resulting in a different set of VIPs. Thus, more generally, one needs to learn the popularity of keys and adapt on the fly.

It is important to note that learning requires some amount of computation and storage. In case of disk-based data structures, this overhead can be relatively small compared to the high latency of accessing storage devices. However, this is not true for hash tables which are typically resident in memory and involve lightweight computation. Even adding a small counter per entry in the hash table can degrade performance considerably, as we show in \S\ref{sec:learning_is_costly}. Thus, the learning mechanisms need to be designed keeping the overhead in check compared to the gains. 

Our contributions in this paper are as follows \textendash\vspace{-1mm}
\begin{enumerate}[leftmargin=*]
	\item \textit{\textbf{Wiscer}} (\S\ref{sec:wiscer}) \textendash\ We developed a configurable tool for measuring the performance of hash tables. Wiscer can be used to generate workloads with varying levels of skew in popularity, with different ratios of fetch, insert and delete operations, and shifting hot set of keys over time. To our knowledge, no existing benchmarking tool captures all of this behavior in one place.\vspace{-1mm}
	
	\item \textbf{\textit{Roofline Analysis of the VIP configuration}} (\S\ref{sec:roofline}) \textendash\ We study the benefit of the VIP configuration (Fig. \ref{fig:vip_arrangement}) given prior knowledge of popularity. Since there is no overhead of learning, this analysis shows the maximum gain one can obtain from adapting to the skew (for a hash table with 10M keys at load factor 0.6, we observe a 57\% increase in throughput from the VIP configuration in the best case).\vspace{-1mm}
	
	\item \textbf{\textit{Learning on a budget}} (\S\ref{sec:learning}) \textendash\ We developed lightweight mechanisms for \textit{learning} the popularity distribution on the fly, \textit{adapting} to the skew, \textit{sensing} changes in the popularity distribution, and \textit{dynamically switching on/off} learning to control the overhead. Put together, they give us the VIP Hashing method for learning the skew in popularity on the fly.\vspace{-1mm}
	
	\item \textbf{\textit{\revision{Application to hash joins}}} \revision{(\S\ref{sec:pk_fk_join})} \textendash\ \revision{We study the application of VIP hashing to PK-FK hash joins, and we obtain a 13-23\% reduction in canonical join query execution time (for a cardinality ratio of 1:16 in the relations and a hash table with load factor of 1.4). \shepherd{We implemented VIP hashing in DuckDB~\cite{duckdb} to speed up PK-FK hash joins in single-threaded mode, and we obtain a net reduction of 20\% in end-to-end execution time of TPC-H query 9~\cite{tpch} under low and medium skew.}} \vspace{-1mm}
	
	\item \textbf{\textit{\revision{Application to point queries}}} \revision{(\S\ref{sec:point_lookups})} \textendash\ \revision{Another common use of hash tables is processing point queries. We test VIP hashing at a load factor of 0.95 under a variety of workloads involving insert and delete operations, shifting popularity distribution, different rates of shift, etc. \shepherd{A gain in throughput of 22\% (77\%) is obtained under low (medium) skew}, while our choice of parameters ensures that the loss due to the overhead of learning is capped in the worst case.}
	
%	\item \textbf{\textit{\revision{Application to point queries}}} (\S\ref{sec:evaluation}) \textendash\ 
%	We test VIP hashing under a variety of workloads generated using Wiscer involving different levels of skew, insert and delete operations, shifting popularity distribution, different rates of shift, etc. We discuss some representative workloads in \S\ref{sec:evaluation} that highlight the robustness of VIP hashing under different conditions.
	
%	\item (Maybe) \textbf{\textit{Applications to query processing}} (\S\ref{sec:applications}) \textendash\ We apply the VIP hashing scheme to hash joins and aggregation on a stream.
\end{enumerate}

%In \S\ref{sec:evaluation}, we run experiments on a hash table with one million keys at a load factor of 0.95. The gain obtained from VIP hashing varies depending upon on a host of factors such as the size of the hash table, number of keys, level of skew in the requests, the request mix in the workload, the parameters of the underlying hardware, etc. In our experiments, we obtain a gain from VIP hashing ranging from 7-19\% depending on the workload, while our choice of parameters ensures that the loss in the worst case due to the overhead of learning is capped (2.5\% loss observed under uniform popularity distribution).

\revision{Overall, our experiments in \S\ref{sec:evaluation} show that the VIP hashing is a fully online non-blocking learning method that captures the skew in popularity on the fly, while being robust to inserts, deletes, and shifting popularity distribution. We discuss related work in \S\ref{sec:related_work} and conclude in \S\ref{sec:conclusion}.}

%Overall, our experiments in \S\ref{sec:evaluation} indicate that the VIP hashing method is robust to different workload conditions, and improves the performance when there is skew in the popularity of keys. 
	
	\section{Background}\label{sec:background}
	\subsection{Hash Tables}\label{sec:background_ht}

A hash table~\cite{hashtable} is an associative data structure that maps keys to values. In our work, we focus on \textit{chained hashing} (hereafter referred to as hash table). A hash table (Fig.~\ref{fig:vip_idea}) uses a \textit{hash function} to map each key to a unique index or \textit{bucket}. Since more than one key can be mapped to the same bucket, the data structure resolves these \textit{collisions} by maintaining a chain (linked list) of entries belonging to the bucket. The flexibility provided by this data structure for performing insert and delete operations, along with variable length keys and values make it a popular choice in many data systems~\cite{redis,mysql,memcached,sqlite}.

\subsubsection{On Properly Configuring the Hash Table}\label{sec:configuring_ht}
 
In this paper, we focus on hashing of 8-byte integer keys and values, which is a well studied problem in past research~\cite{richter,spyros}. It is important to configure the hash table correctly to draw reliable conclusions, and there are two important factors to consider. The first is the choice of the hash function. In our work, we use \textit{MurmurHash}~\cite{murmurhash}, which is a strong hash function that provides good collision resistance in practice. The second critical aspect is the \textit{load factor}, which is the ratio of keys to the number of buckets in the hash table. Higher load factors correspond to fewer buckets, which lead to longer chains on an average, whereas lower load factors require more buckets and consume more memory. \revision{Informed by parameter choices in popular open-source systems~\cite{redis,memcached,postgres}, 	we maintain a load factor between \revision{$0.5$ and $1.5$}} to ensure that collisions are at an acceptable level while utilizing memory efficiently. Wherever applicable, we \textit{rehash} the hash table to maintain this range of load factor. The number of buckets in the hash table are set to be a power of two, which is a common choice~\cite{postgres,memcached,redisBlog} that speeds up the computation of the hash function. If the load factor exceeds 1.5 (falls under 0.5), we double (half) the number of buckets in the hash table.

\begin{table*}[]
	\caption{\textbf{Configuration options supported by \textit{Wiscer}}}
	\label{tab:wiscer_options}
	\centering
	\begin{tabular}{|c|l|}
		\hline
		\textbf{Option}                                                       & \multicolumn{1}{c|}{\textbf{Description}}                                                         \\ \hline
		\textit{zipf}                                                                       & The zipfian factor of the popularity distribution. \textit{zipf} = 0 corresponds to uniform popularity.    \\ \hline
		\textit{initialSize}                                                                & Initial number of keys in the hash table before running any operations.                  \\ \hline
		\textit{operationCount}                                                             & Total number of operations (fetch, inserts, etc.) to run on the hash table.                       \\ \hline
		\textit{\begin{tabular}[c]{@{}c@{}}(fetch/insert/delete)\\ Proportion\end{tabular}} & Proportion of operations that are fetch/insert/delete.                                            \\ \hline
		\textit{distShiftFreq}                                                              & A shift in popularity distribution occurs after every \textit{distShiftFreq} operations.                   \\ \hline
		\textit{distShiftPrct}                                                              & The popularity distribution shifts by \textit{distShiftPrct}\% every \textit{distShiftFreq} operations.           \\ \hline
		\textit{\revision{storageEngine}}                                                              & \begin{tabular}[c]{@{}l@{}}\revision{Which storage engine to benchmark. Options are} \\ \revision{\textit{ChainedHashing} (default), \textit{VIPHashing}, and \textit{none} (store workload to disk).}\end{tabular} \\ \hline
		
		\textit{\revision{keyPattern}}                                                                & \revision{The pattern of keys to generate \textendash\ \textit{random} (default) or \textit{sequential} ($1$ to $n$).}                 \\ \hline
		
		\textit{keyOrder}                                                              & \begin{tabular}[c]{@{}l@{}}The popularity rank of keys relative to the insertion order. Options are\\ \textit{random} (default) and \textit{sorted} (where keys are inserted in increasing order of popularity; a.k.a. latest).\end{tabular} \\ \hline
		
		 \textit{randomSeed}
		 & \begin{tabular}[c]{@{}l@{}}The seed value (unsigned integer) to initialize the random number generator (default = 0).\\ The random number generator is used to populate the hash table and generate the workload.\\
		 Different seed values result in different instances of keys and the workload.
		\end{tabular} \\ \hline
	\end{tabular}
\end{table*}
\begin{figure}
	\centering
	\includegraphics[scale=0.42]{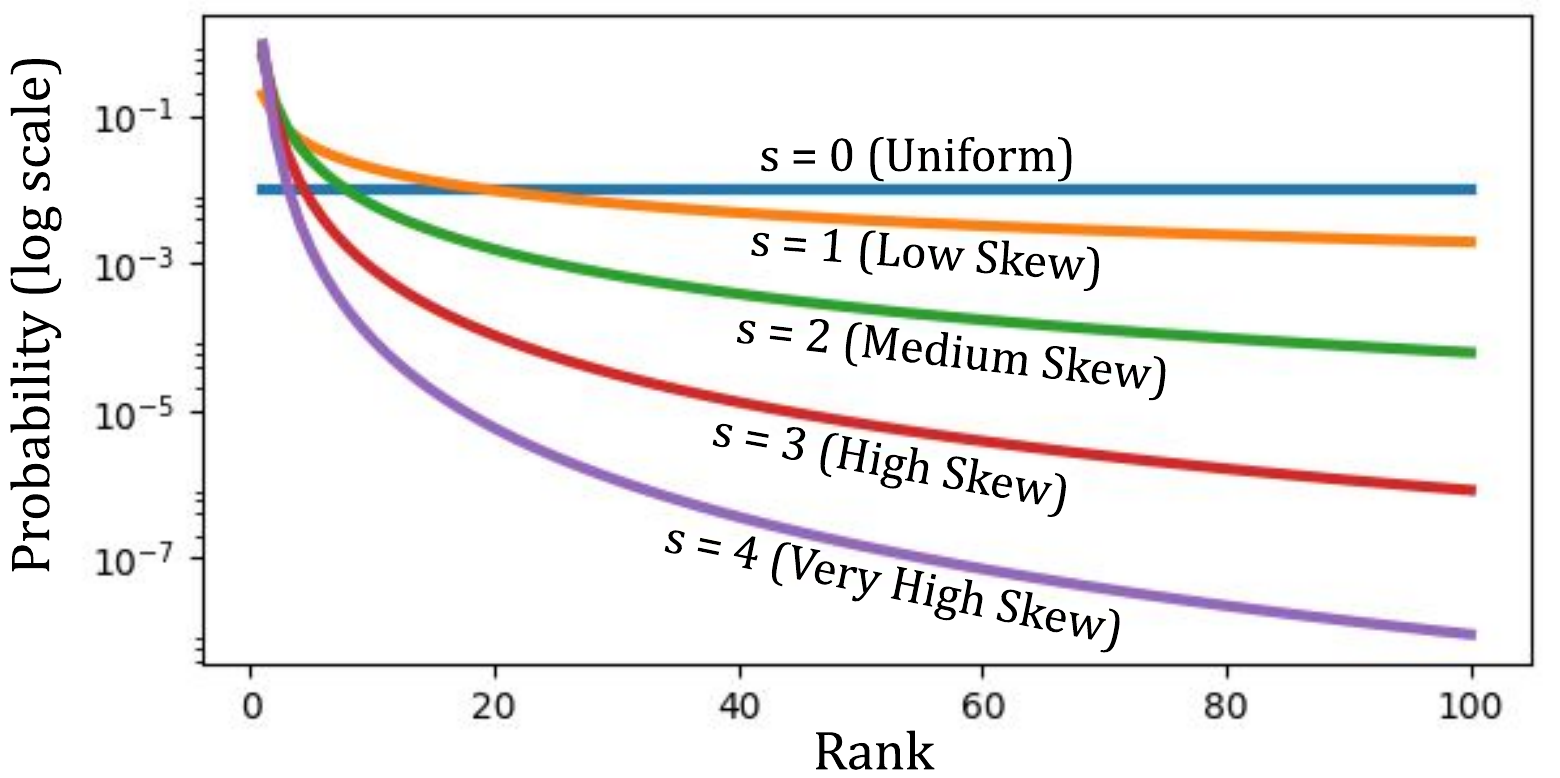}
	\caption{\textbf{Popularity distribution of keys (number of keys $N=100$) for different Zipfian skew factors $s$.}}
	\label{fig:zipf_factors}
\end{figure}

%In our experiments, we use a strong hash function, namely \textit{Murmur hash}~\cite{}. Another important aspect to consider is the \textit{load factor}, i.e, the ratio of keys in the hash table to the number of buckets.
%It is important to configure the hash table such that collisions are at an acceptable level, otherwise it is possible to draw incorrect conclusions. Very high load factors correspond to fewer buckets, which can lead to longer chains and more collisions. Typically, a load factor between 0.5 and 1 combined with a strong hash function mitigates collisions.

\subsection{Some Probability Bounds and Theorems}\label{sec:background_prob}

%In this section, we discuss some theorems and bounds related to probabilistic random variables. We use these in our work to generate workloads with skewed popularity distributions, and to develop different mechanisms used by VIP hashing.

\revision{Below we discuss some tools related to probabilistic random variables that we use in our work.}

\begin{itemize}[leftmargin=*]
	\item \revision{\textbf{Zipfian distribution}: We use Zipfian distribution~\cite{zipf} to model varying levels of skew in fetch operations issued to keys in a hash table. Zipfian distribution has been adopted by multiple studies in the past~\cite{anna,spyros,fbstudy} to statistically model skew in popularity, as it captures the power law~\cite{powerlaw} characteristics of workloads that are often observed in practice~\cite{fbstudy,zipf2}.}
	\item \revision{\textbf{Estimating mean and variance}: Let $X$ be a random variable with mean $\mu$ and variance $\sigma^2$. Let $X_1,\ X_2,...\ X_n$ be $n$ independent and identically distributed (i.i.d.) measurements of $X$. The estimated mean $\hat{\mu}$ and estimated variance $\hat{\sigma}^2$ can be evaluated as
	\begin{displaymath}
	\hat{\mu} = \frac{\sum\limits_{i=1}^{n}X_i}{n},\hspace{4mm}
	\hat{\sigma}^2 = \Bigg(\frac{\sum\limits_{i=1}^{n}X_i^2}{n-1} - \frac{\Big(\sum\limits_{i=1}^{n}X_i\Big)^2}{n(n-1)}\Bigg)
	\end{displaymath}}
	\item \revision{\textbf{Gaussian tail bound confidence interval}: For a random variable $X$ (refer above), the central limit theorem (CLT)~\cite{clt} states that the error in estimated mean $(\hat{\mu}-\mu)$ is approximately Gaussian distributed $\mathcal{N}(0, \sfrac{\sigma^2}{n})$. By applying the Gaussian pdf, a confidence interval can be obtained for the error $(\hat{\mu}-\mu)$ as follows 
	\begin{displaymath}
	P(|\hat{\mu} - \mu| \leq t) \geq \bigg(1-exp\bigg(\frac{-nt^2}{2\sigma^2}\bigg)\bigg) = \frac{L}{100}
	\end{displaymath}
	Thus, we can at least be $L\%$ confident that the error $|\hat{\mu}-\mu|$ is less than $t$. Note that the confidence increases exponentially with $n$ (number of samples $X_i$ drawn). It is important to note that $(\hat{\mu}-\mu)$ is only \textit{approximately} Gaussian, so the confidence interval obtained from applying Gaussian tail bound is a heuristic.}
\end{itemize}

	\section{Skewed Workload Generation with \secit{WISCER}}\label{sec:wiscer}
	%We talk about supporting zipfian distribution and simulating shifting hotsets.

%In this section, we first give an overview of Wiscer (\S\ref{sec:wiser_overview}) followed by describing the workload suite (\S\ref{sec:workload_suite}) that we use to evaluate VIP hashing.
\subsection{Overview}\label{sec:wiser_overview}
Wiscer~\cite{wiscer_bib} is a workload generation tool that we propose in this paper. Wiscer has multiple configuration options (Table \ref{tab:wiscer_options}) that can be used to generate workloads with different levels of skew, varying proportions of fetch, insert, delete operations, different rates of popularity shift, etc. Below are some key features of Wiscer:

\begin{itemize}[leftmargin=*]
	\item \textbf{Level of skew}: Increasing levels of skew in the popularity distribution can be simulated by increasing the \textit{zipf} factor. For instance, $zipf=0$ and $zipf=4$ correspond to uniform distribution and very high skew respectively (see Fig.~\ref{fig:zipf_factors}).
	\item \textbf{Simulating popularity distribution shift}: The two related configuration options are \textit{distShiftFreq} and \textit{distShiftPrct}. After every \textit{distShiftFreq} fetch operations, the topmost popular keys that constitute \textit{distShiftPrct} of the requests are randomly replaced by less popular keys. This simulates a behavior where keys in the hot set become less popular after some time, which has also been observed in some real-world workloads~\cite{fbstudy}.
	\item \revision{\textbf{Benchmarking hash table implementations}: Wiscer can optionally be used to compare different hash table implementations (option \textit{StorageEngine}) to directly process the generated workloads without intermediate storage.}
	\item \revision{\textbf{Fine-grained performance metrics using hardware counters}: When using Wiscer for benchmarking, operations are issued to the configured  hash table in batches of one million requests at a time, and fine-grained metrics are collected per batch. Wiscer uses hardware counters provided by the Intel's Performance Monitoring Unit (PMU)~\cite{pmu} to get low-level performance metrics such as cache misses, number of cycles, retired instructions, etc.}
	
%Thus, we can obtain both the average and the fine-grained statistics of throughput using Wiscer.
%	\item \textbf{Minimizing interference} \textendash\ Wiscer minimizes interference caused to benchmarking in two ways. First, the workload is generated in bulk at the start before issuing any requests to the storage engine. Secondly, the prefetch length is controlled to minimize the amount of cache occupied by the workload.
%	\item \textbf{Rehashing to maintain reasonable load factor} \textendash\ 
%	\item \textbf{Hardware counters} \textendash\ Wiscer uses hardware counters provided by the Intel's Performance Monitoring Unit~\cite{pmu} (PMU) to get low-level performance metrics such as cache misses, number of cycles, retired instructions, etc. for every batch of requests.
\end{itemize}

%Thus, Wiscer allows us to compare VIP hashing to the vanilla implementation using a myriad of workloads while measuring multiple metrics to reason about the performance.
% I have to mention somewhere about rehashing to maintain good load factor. Otherwise it is important to draw incorrect conclusions if the load factor increases arbitrarily.

\subsection{Experimental Configuration}\label{sec:experimental_configuration}
All experiments in this paper are run on a Cloudlab~\cite{cloudlab} machine with two 10-core Intel Xeon Silver 4114 CPUs with a peak frequency of 3.0GHz. The benchmarking process is pinned to a single core to avoid any overhead of context switching. The CPU scaling governor of the core has been set to \textit{performance}, thus fixing the frequency to 3.0GHz at all times. The CPU has an L3 cache of 13.75MB, and the server machine has 192GB of RAM. This CPU belongs to the Skylake Intel architecture family~\cite{cpu}, and the PMU's hardware counters are programmed accordingly. The server machine is used exclusively for running Wiscer to mitigate interference from any concurrent processes.

	\section{Roofline Study}\label{sec:roofline}
	\begin{figure}
	\centering
	\subfloat[Default configuration. A total displacement of 12 (=2$\times$(2+1+3)) is required to process the fetch requests. The less popular keys in the path of popular keys need to accessed as well.]{
		\includegraphics[scale=0.26]{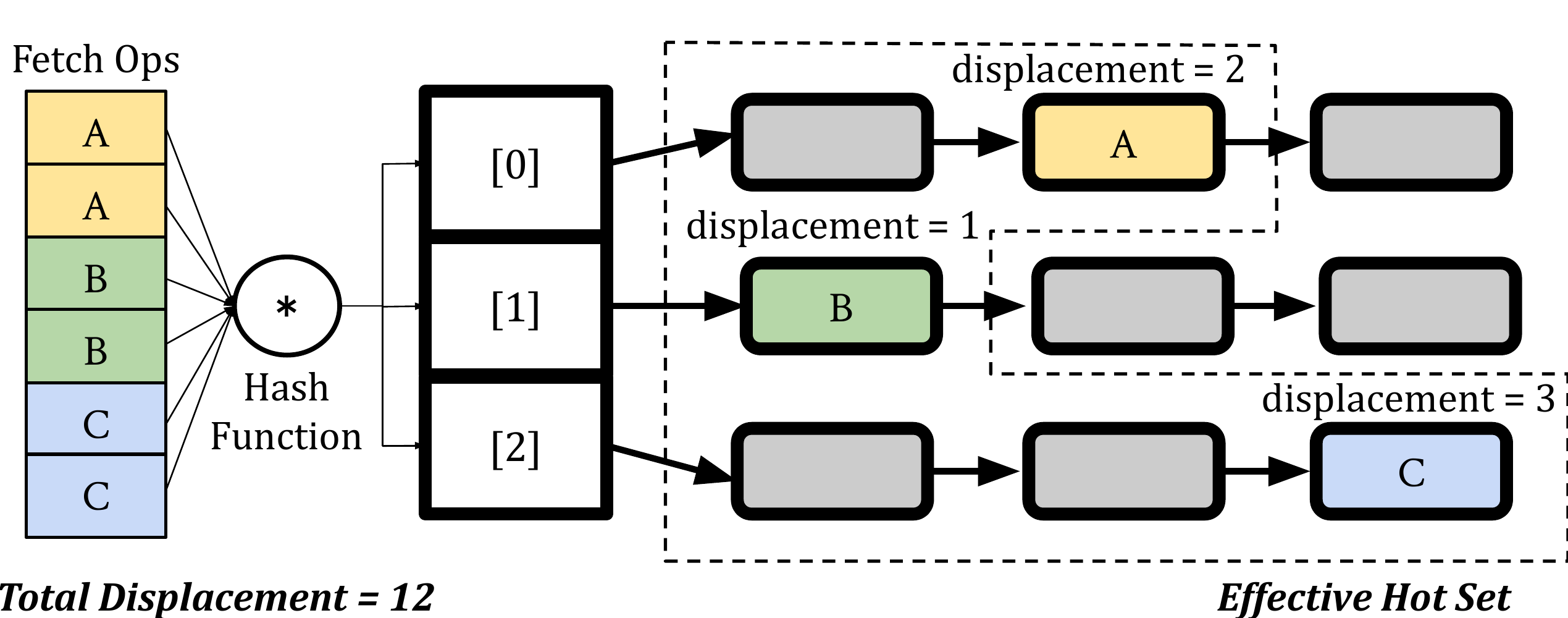}
		\label{fig:disp_random_arrangement}
	}
	\vspace{3mm}
	\subfloat[VIP configuration. A total displacement of 6 (=2$\times$(1+1+1)) is required to process the fetch requests. Only the popular keys are accessed.]{
		\includegraphics[scale=0.26]{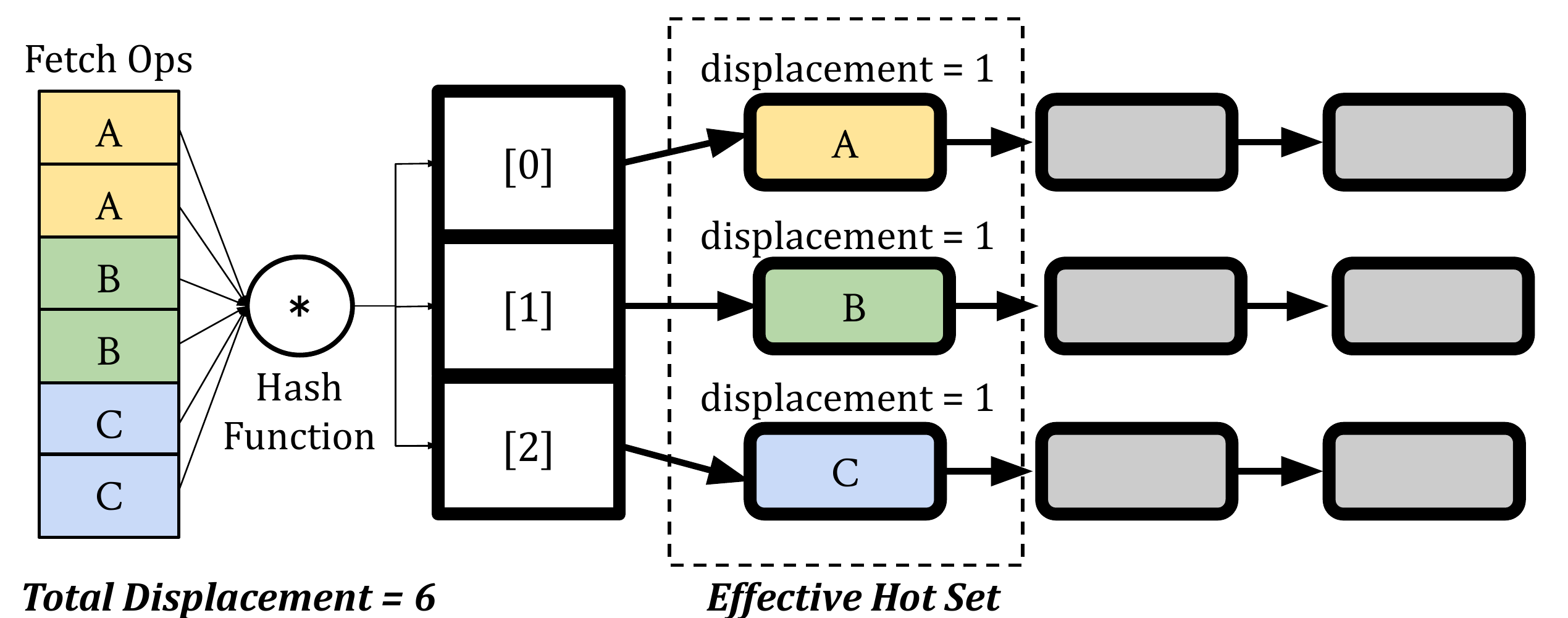}
		\label{fig:disp_vip_arrangement}
	}
	%	\vspace{-0.4em}
	\caption{\textbf{Processing fetch requests in the Default vs the VIP configuration. Unpopular keys have been grayed out. The total \textit{displacement} (number of keys accessed) is higher in the Default configuration requiring more pointer dereferences. Also, the effective hot set is larger, increasing the likelihood of cache misses relative to the VIP configuration.}}
	\label{fig:disp_analysis}
\end{figure}

\begin{figure*}
	\centering
	\subfloat[\revision{Fetch operation throughput}]{
		\includegraphics[scale=0.34]{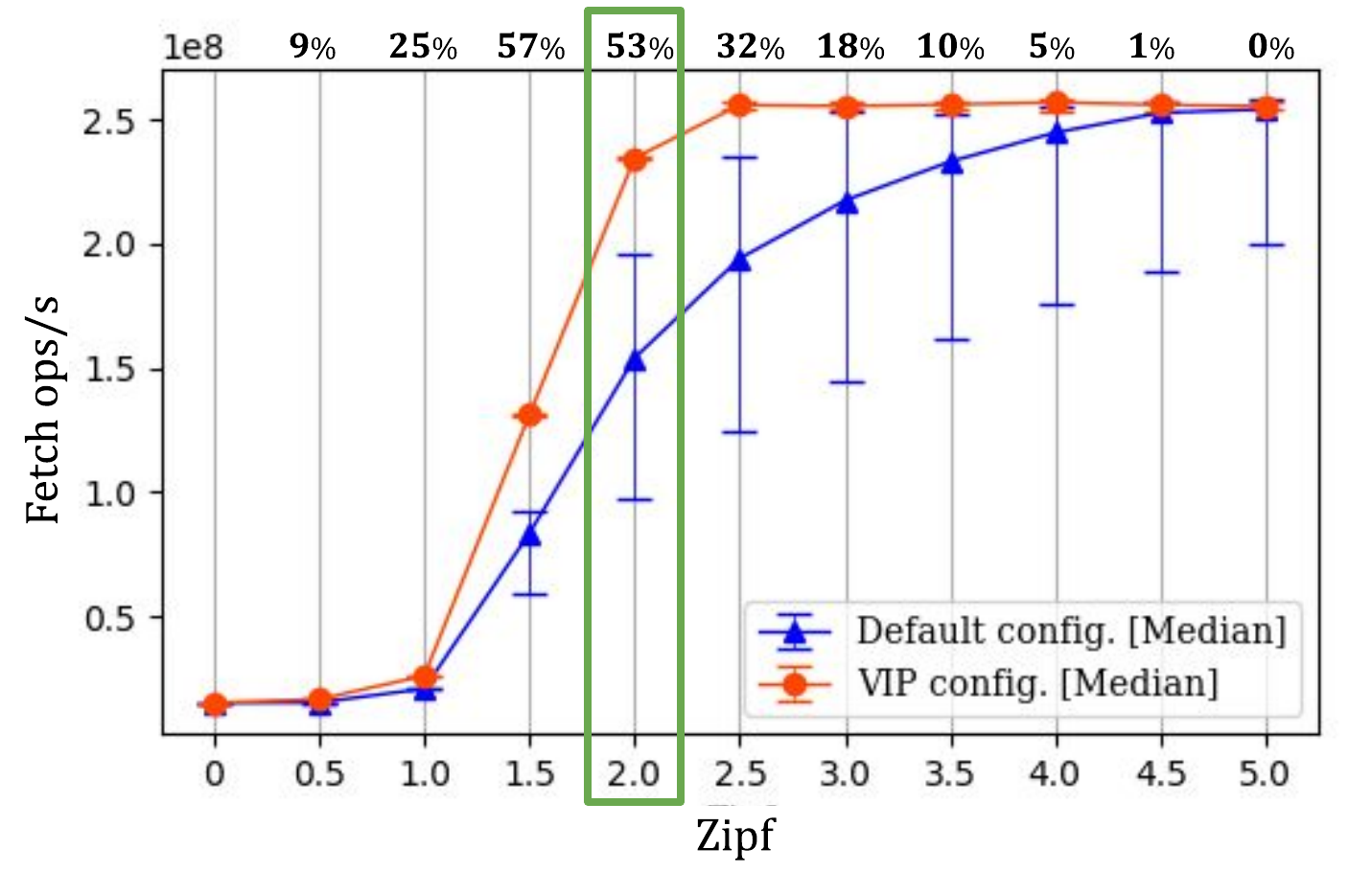}
		\label{fig:roofline_throughput}
	}\hspace{5mm}
	\subfloat[\revision{Displacement}]{
		\includegraphics[scale=0.34]{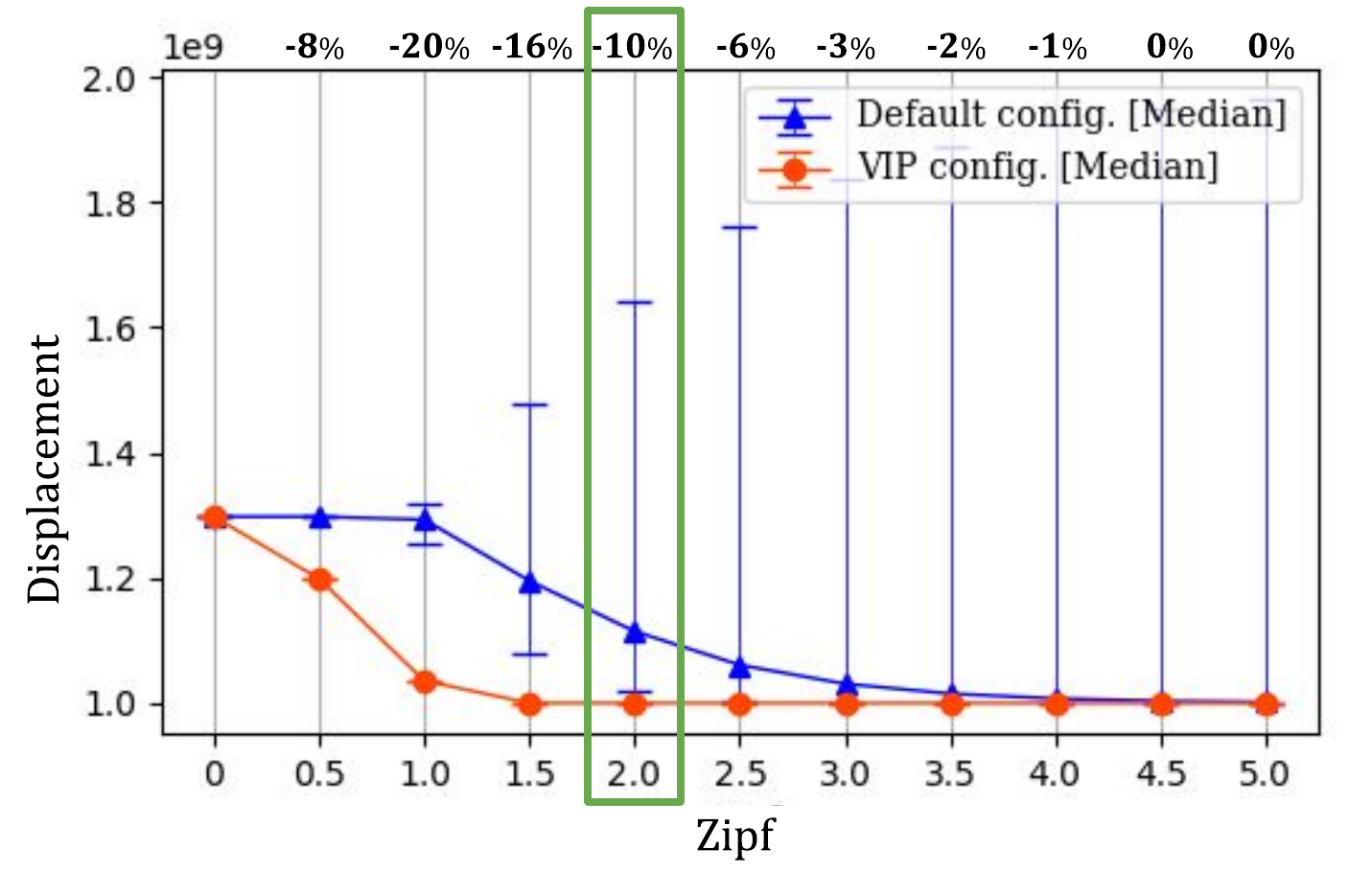}
		\label{fig:roofline_disp}
	}\hspace{5mm}
	\subfloat[\revision{Retired instructions}]{
		\includegraphics[scale=0.34]{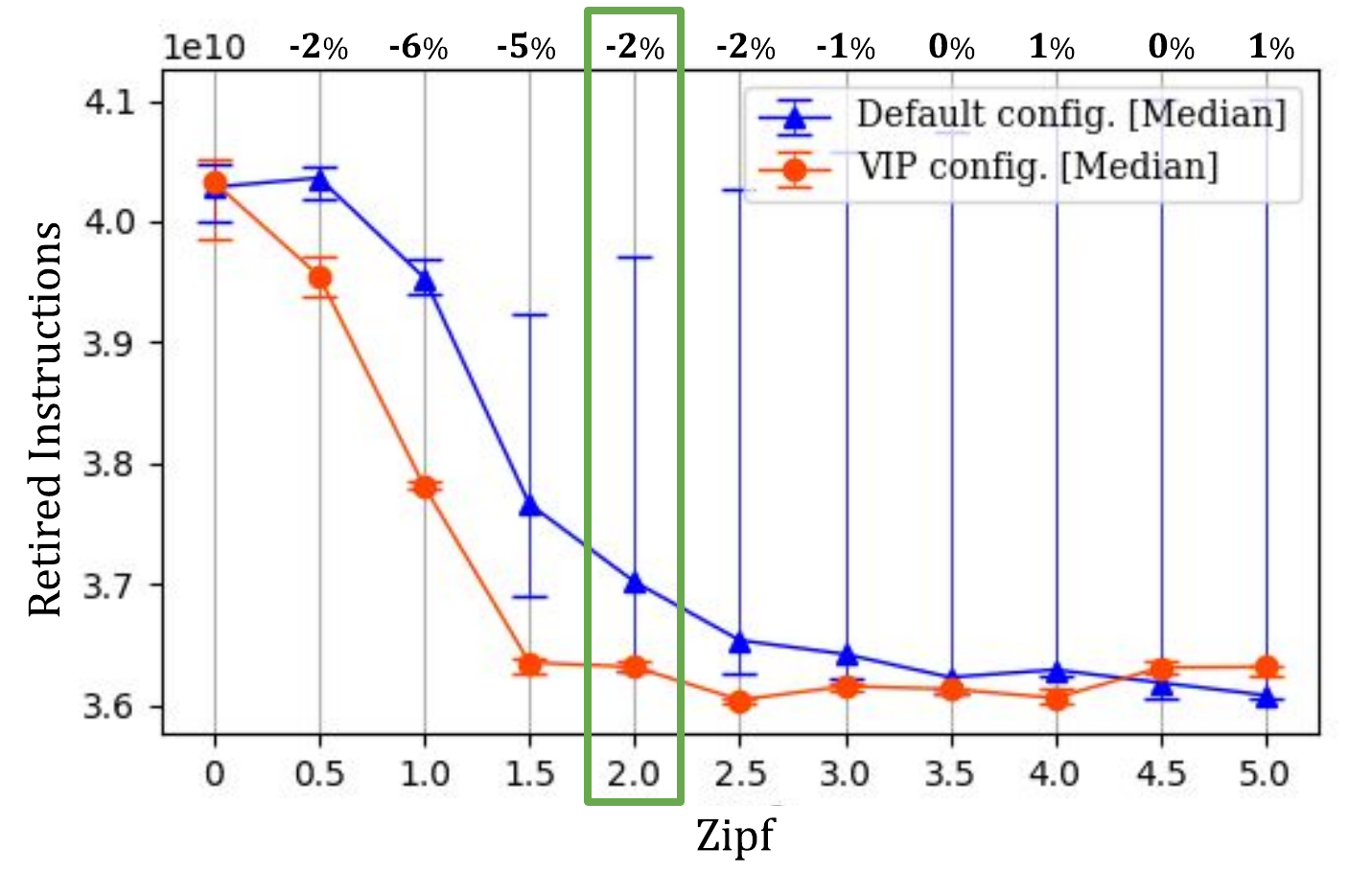}
		\label{fig:roofline_retd_insts}
	}
	%	\vspace{1mm}
	
	\subfloat[\revision{L1 cache misses}]{
		\includegraphics[scale=0.34]{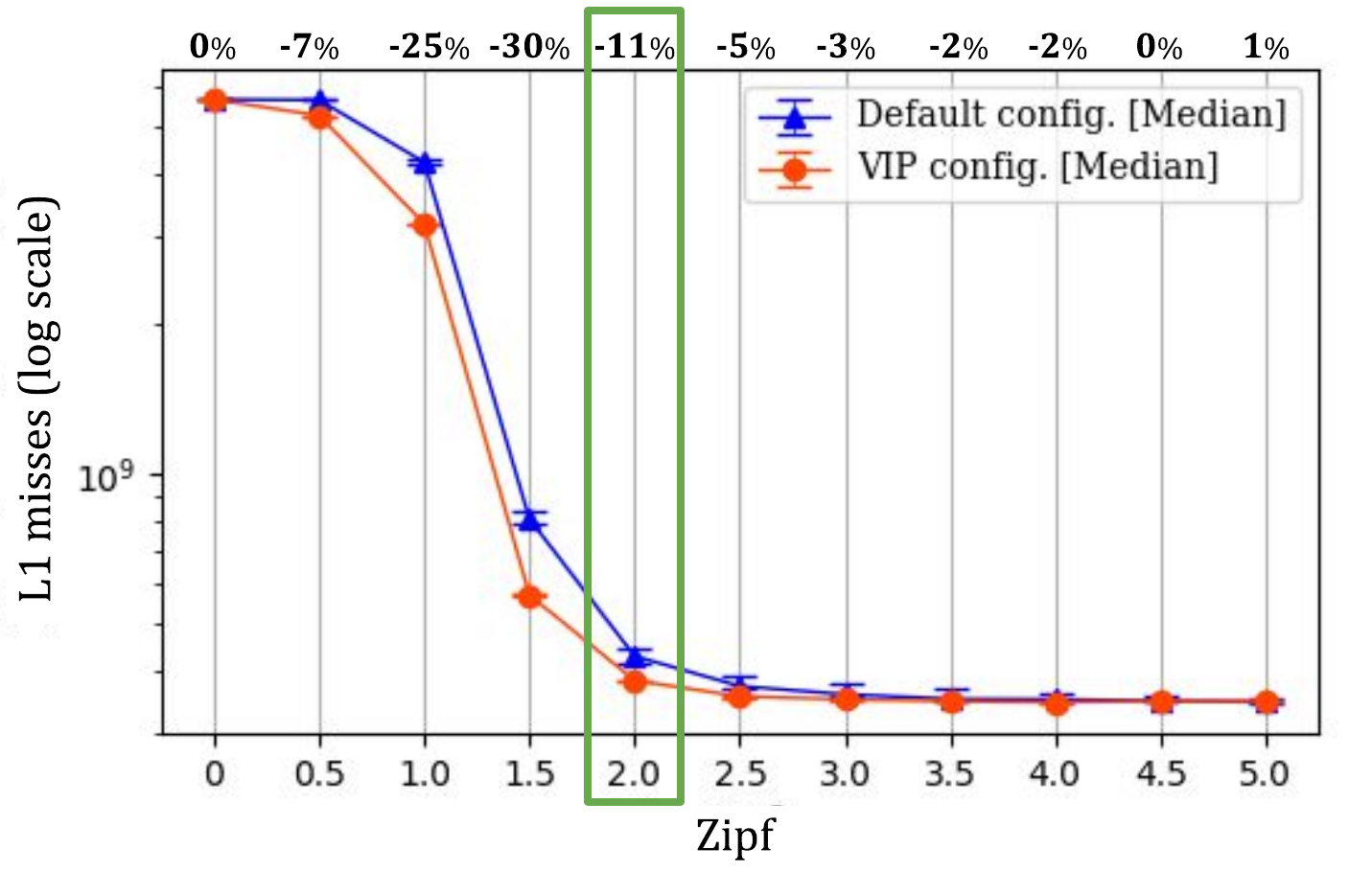}
		\label{fig:roofline_l1_misses}
	}\hspace{5mm}
	\subfloat[\revision{L2 cache misses}]{
		\includegraphics[scale=0.34]{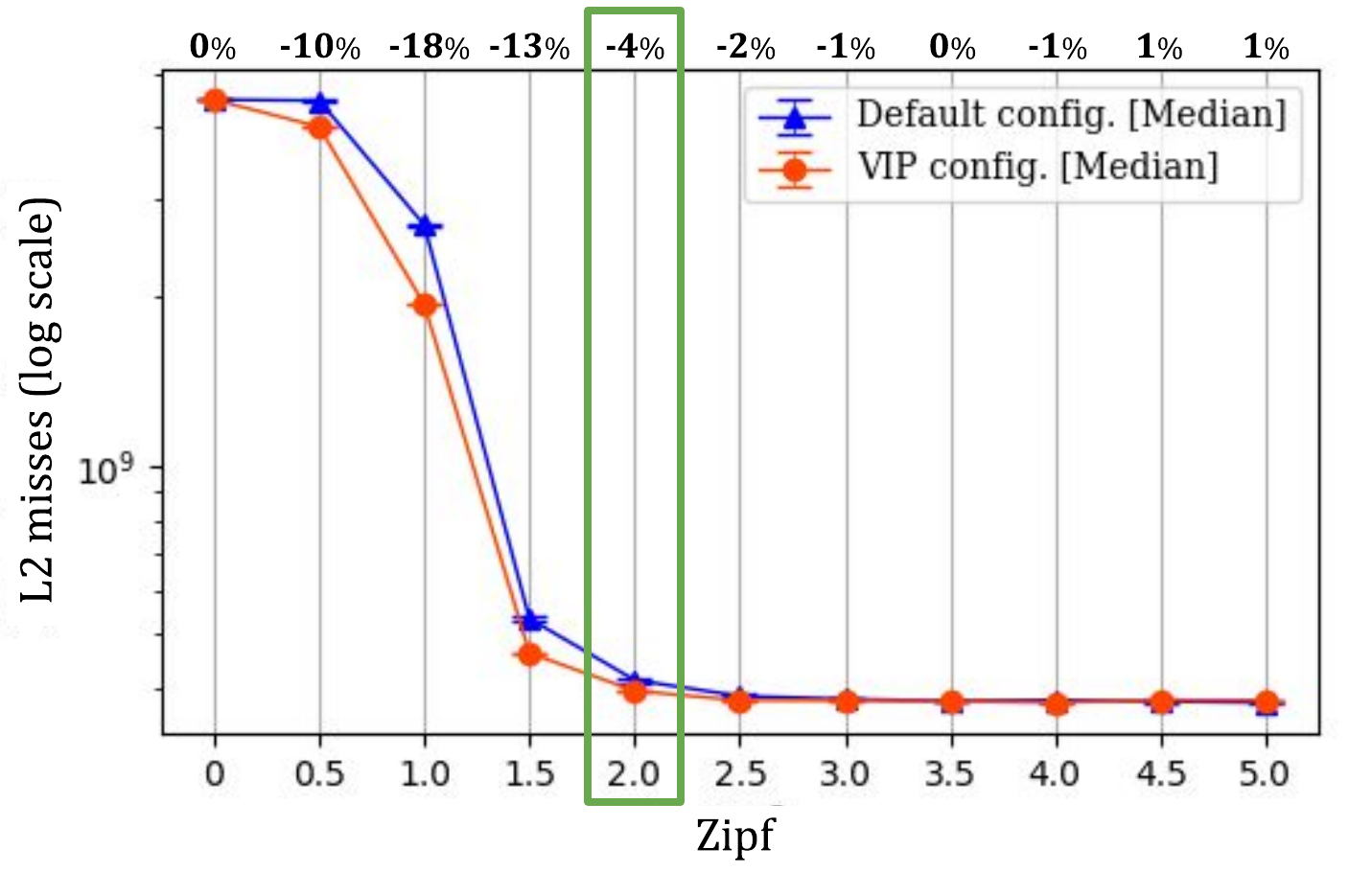}
		\label{fig:roofline_l2_misses}
	}\hspace{5mm}
	\subfloat[\revision{L3 cache misses}]{
		\includegraphics[scale=0.34]{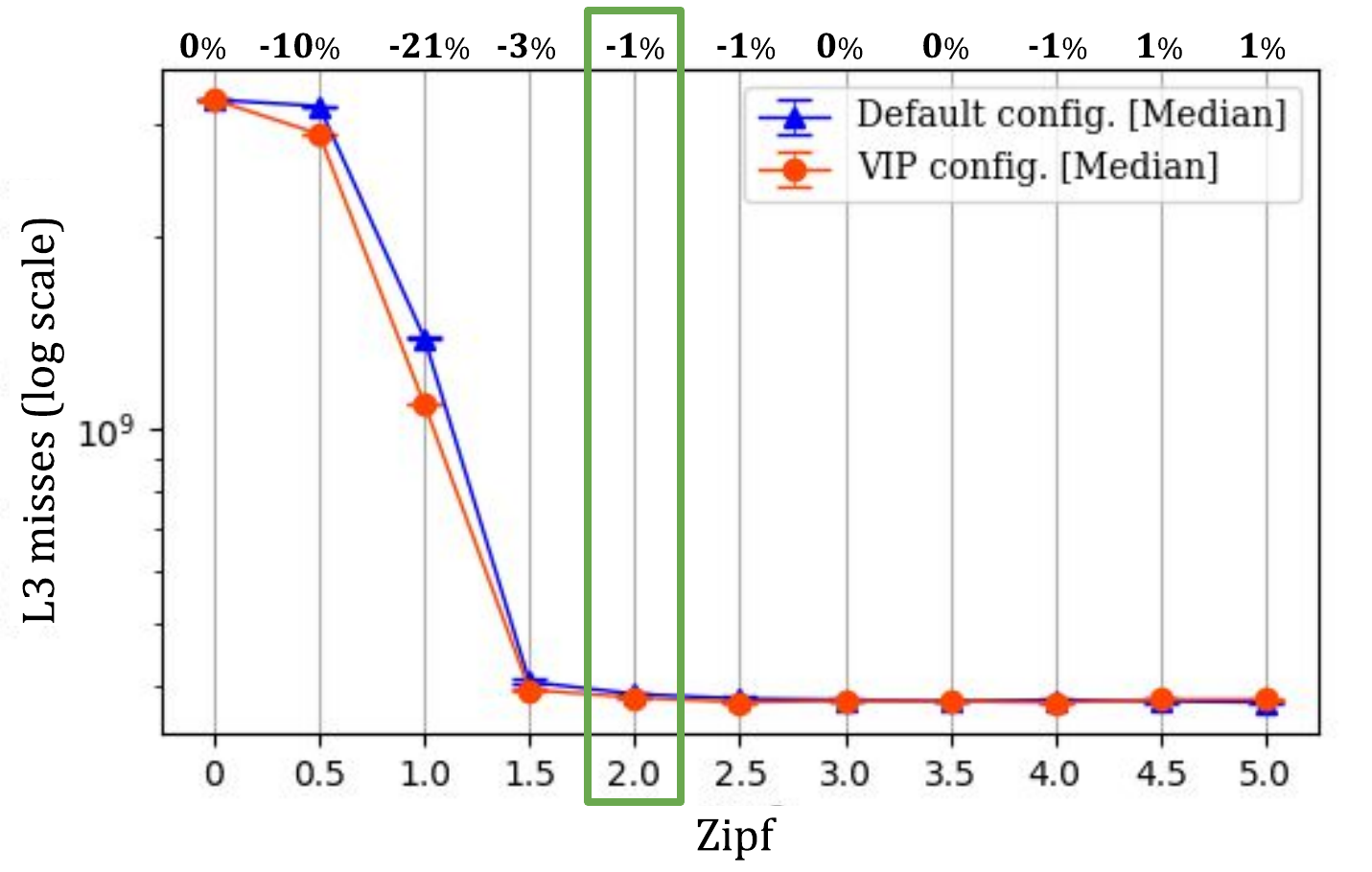}
		\label{fig:roofline_l3_misses}
	}
	%	\vspace{5mm}
	
%	\subfloat[Retired instructions]{
%		\includegraphics[scale=0.6]{figures/roofline_retd_insts.pdf}
%		\label{fig:roofline_retd_insts}
%	}
%	\vspace{-0.4em}
	\caption{\textbf{Relative performance of the VIP vs the Default configurations as the skew in popularity increases. One billion fetch requests are issued to a hash table with 10M keys (load factor $0.6$) for varying levels of skew from $zipf=0$ to $zipf=5$. Each reported data point is the \revision{median} over 10 runs with different random seeds. Percentage difference indicated at the top of each plot is the difference between median metrics of the VIP vs the Default configuration. \revision{The gain in fetch operation throughput varies with skew}, and we obtain 53\% increase in throughput for medium skew ($zipf = 2.0$). Lesser number of cache misses and instructions executed contribute to the gain obtained from the VIP configuration. Results are discussed in} \S\ref{sec:roofline_results}\textbf{.}}
	\label{fig:roofline_results}
%	\vspace{-1.5mm}
\end{figure*}

%In this section, we show the roofline benefits of exploiting popularity. We show the performance gap between hashmap configurations where entries are in random order vs sorted order of popularity.
In this section, we compare the performance of the Default and VIP configurations when the popularity of keys is static and known in advance. Since there is no overhead of learning involved in this case, this roofline study shows the maximum gain one can get from the VIP configuration for \revision{different levels of skew (\S\ref{sec:roofline_increasing_skew}) in popularity at different load factors (\S\ref{sec:roofline_increasing_lf}) of the hash table}.

\subsection{\revision{Default vs VIP Configuration}}
\subsubsection{Motivation}

Fig. \ref{fig:disp_analysis} shows an example of processing fetch requests in the Default and the VIP configurations. A key parameter to note is the \textit{displacement} encountered, which is the total number of keys that were accessed to process the fetch requests. Accessing a key requires dereferencing a pointer
and some additional computation. The displacement encountered in the Default configuration is higher as the less popular keys in the path to VIPs need to be accessed when processing the fetch requests and effectively become part of the hot set. A larger hot set increases the likelihood of cache misses, and we observe this trend in our experiments described next.

\subsubsection{\revision{Generating the configurations using Wiscer}}

In the VIP configuration, keys in the hash table are arranged in descending order of popularity in the bucket chains (see Fig. \ref{fig:disp_vip_arrangement}). We attain this configuration by running Wiscer with the default storage engine (\textit{ChainedHashing}) and inserting keys in increasing order of popularity (\textit{keyorder=sorted}, default is \textit{random}). Insert operations on the hash table are performed at the front of the bucket chain (\S\ref{sec:background_ht}). Thus, when inserting keys in the sorted order, entries are automatically placed in decreasing order of popularity as more popular keys are inserted later and are ahead in the bucket chain. The Default configuration is generated using the default parameters of Wiscer.

\subsection{\revision{Impact of Increasing Skew}}\label{sec:roofline_increasing_skew}
\subsubsection{Workload}\label{sec:roofline_increasing_skew_workload}
We compare the throughput of fetch operations in the Default and VIP configurations. We use Wiscer (Table~\ref{tab:wiscer_options}) to generate fetch requests with increasing levels of skew (\textit{zipf} = 0 to 5 in steps of 0.5) which are issued to a hash table with 10 million keys at a load factor of $0.6$ ($=10^7/2^{24}$). For each level of skew and hash table configuration, Wiscer is run with 10 distinct random seed values to populate the hash table and generate the workload. Each random seed results in a different arrangement of keys in the hash table. The popularity distribution is static, i.e., the rank of the keys remains the same throughout a run. One billion fetch requests are issued to the hash table for each random seed, and the data points reported in Fig. \ref{fig:roofline_results} are the median statistics over the 10 runs. We have run experiments on smaller (1M entries) and larger (100M entries) hash tables and found the trends to be similar.
% Fig. \ref{fig:roofline_results} shows the results averaged over the 10 runs for each configuration and zipf value.

\subsubsection{Results}\label{sec:roofline_results}
The results of this experiment are shown in Fig. \ref{fig:roofline_results}. \revision{The gain in throughput ranges from 9\%-57\% depending upon the level of skew in popularity}. Below we discuss our takeaways from the performance metrics measured using Wiscer:

\begin{itemize}[leftmargin=*,noitemsep]
\vspace{-0.5mm}
\item \textit{Throughput}: The gap in performance between the VIP and the Default configuration increases up to $zipf=2$ (medium skew), and gradually diminishes as the skew becomes very high ($zipf=4.5$ or $5$). This behavior is correlated with the hot sets becoming smaller as the skew increases and progressively becoming (L1/2/3) cache resident at different rates for the two configurations.\vspace{2mm}

\item \textit{Displacement}: As expected, the displacement encountered in the VIP configuration is lower than the Default (see Fig.~\ref{fig:disp_analysis}). For $zipf=1.5$ and up, the total displacement becomes close to 1B (for 1B fetch requests), indicating that popular keys are at the front of their chains (displacement $=1$) in the VIP configuration. \revision{For the Default configuration, the median displacement approaches 1B at higher levels of skew ($zipf\geq4$), but the variance is high as some random seeds can result in the popular keys placed further in the chains (however the likelihood of this happening is low as the load factor is not very high).}\vspace{2mm}

\item \textit{Instructions Executed}: The instructions executed are lower in the VIP configuration (up to 6\% lower in the best case). The relative trend observed is similar to that of displacement, as the number of instructions executed is correlated with the number of keys accessed.\vspace{2mm}

\item\textit{Cache misses}: The VIP configuration becomes L3 and L1 cache resident (at $zipf=2$ and $2.5$ respectively) more quickly compared to the Default configuration (at $zipf=3.0$ and $4.5$ respectively), which is expected as the hot set of the former is smaller than the latter (Fig.~\ref{fig:disp_analysis}). At very high skew ($zipf=4.5$ and $5$), both the configurations are L1 resident and correspondingly, we do not observe much difference in the throughput. This indicates that caching has a big impact on the performance of hash tables.
\end{itemize}

Overall, we note that since the hot set of the VIP configuration is smaller than the Default, we encounter lower cache misses at all levels of cache. This contributes to the gain in performance we obtain from the VIP configuration.

Another important observation we make is that displacement indicates the goodness of the hash table configuration. The VIP configuration has lower displacement than the Default in all cases (in fact, the VIP configuration has the lowest possible displacement for a given data set, hash table size, hash function, and request skew; we discuss this in \S\ref{sec:vip_hashing_sensing_toggle}). We use this metric in building the mechanisms for sensing and dynamically switching-on/off learning (\S\ref{sec:vip_hashing_sensing_toggle}).

\subsection{\revision{Impact of Increasing Load Factor}}\label{sec:roofline_increasing_lf}
\subsubsection{\revision{Workload}}\label{sec:roofline_increasing_lf_workload}
\revision{In this experiment, we increase the load factor while holding the size of the hash table constant.
Similar to \S\ref{sec:roofline_increasing_skew_workload}, we run one billion fetch operations on a hash table with $2^{24}$ buckets while varying the load factor from $0.5$ to $1.5$ in steps of $0.25$ (this is achieved by increasing \textit{initialSize} from $2^{23}$ to $3\cdot 2^{23}$). Each configuration is run with 10 distinct random seeds and we compare the median statistics over the 10 runs.}

\subsubsection{\revision{Results}}
\revision{Fig.~\ref{fig:roofline_lf} shows the median gain obtained as we increase the load factor from 0.5 to 1.5. We obtain 1.6x, 2.6x, and 1.8x higher throughput from the VIP configuration at low ($zipf = 1$), medium ($zipf = 2$), and high skew ($zipf=3$) respectively at load factor 1.5. In all cases, the gain from the VIP configuration increases as the load factor increases, which is expected as the likelihood of collisions is higher when more keys are present in the hash table. We find that the performance metrics of the VIP configuration are mostly stable (refer to Table~\ref{tab:roofline_lf_metrics}) indicating a stable hot set size, while the performance of the Default configuration becomes steadily worse as the effective hot set grows larger with the load factor.}

\begin{figure}
	\centering
	\includegraphics[scale=0.5]{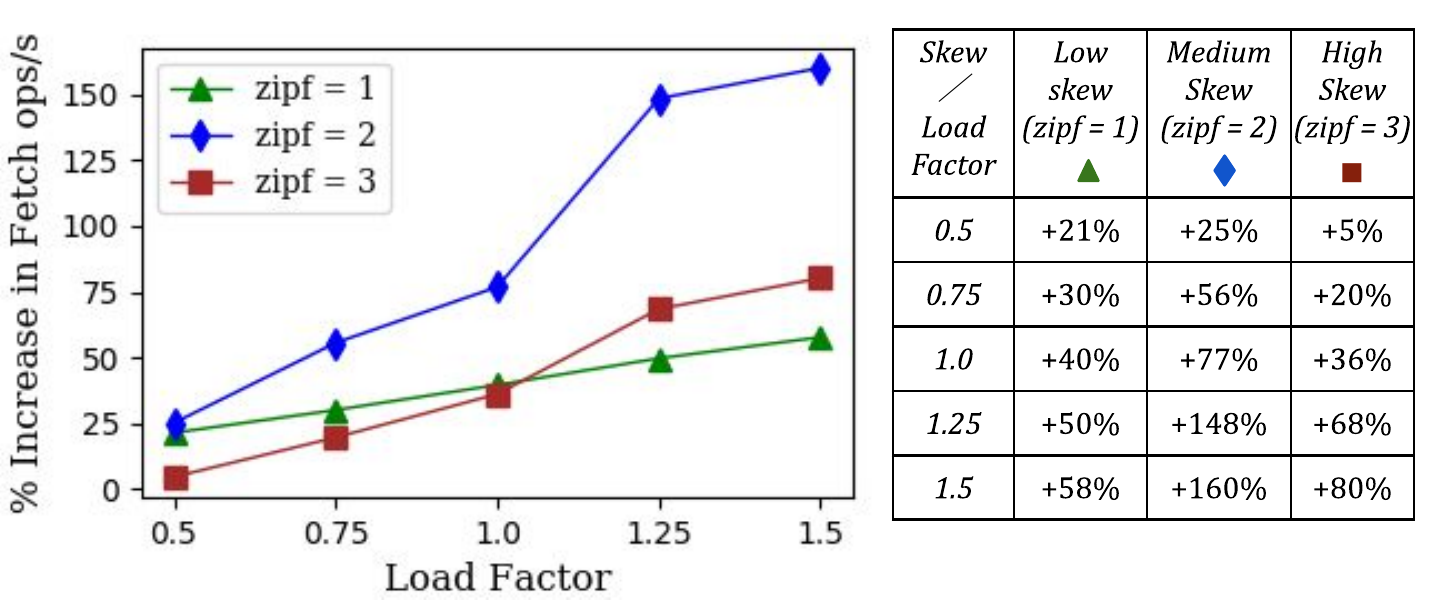}
	\caption{\revision{\textbf{Roofline gain in operation throughput from the VIP vs the Default configuration as the load factor increases. While keeping the number of buckets fixed at $2^{24}$, we increase the load factor from $0.5$ to $1.5$. The performance gain obtained from the VIP configuration increases with the load factor, and can be as high as 160\% (2.6x) for medium skew ($zipf=2$) at load factor $1.5$.}}}
	\label{fig:roofline_lf}
	\vspace{6mm}
\end{figure}

%\begin{table}[]
%	\begin{tabular}{|l|l|l|l|l|}
%		\hline
%		& \multicolumn{1}{c|}{\begin{tabular}[c]{@{}c@{}}Throughput \\ (ops/s)\end{tabular}} & \multicolumn{1}{c|}{\begin{tabular}[c]{@{}c@{}}Avg. \\ Displacement\end{tabular}} & \multicolumn{1}{c|}{L3 Misses}                                  & \multicolumn{1}{c|}{L1 Misses}                                 \\ \hline
%		0.5 & \begin{tabular}[c]{@{}l@{}}235M vs 188M\\ (+25\%)\end{tabular}                     & \begin{tabular}[c]{@{}l@{}}1.0 vs 1.03\\ (-3\%)\end{tabular}                      & \begin{tabular}[c]{@{}l@{}}378M vs 385M\\ (1.8\%)\end{tabular}  & \begin{tabular}[c]{@{}l@{}}380M vs 412M\\ (-8\%)\end{tabular}  \\ \hline
%		1   & \begin{tabular}[c]{@{}l@{}}236M vs 134M\\ (+77\%)\end{tabular}                     & \begin{tabular}[c]{@{}l@{}}1.0 vs 1.17\\ (-15\%)\end{tabular}                     & \begin{tabular}[c]{@{}l@{}}376M vs 387M\\ (-2.6\%)\end{tabular} & \begin{tabular}[c]{@{}l@{}}380M vs 436M\\ (-13\%)\end{tabular} \\ \hline
%		1.5 & \begin{tabular}[c]{@{}l@{}}236M vs 90M\\ (+160\%)\end{tabular}                     & \begin{tabular}[c]{@{}l@{}}1.0 vs 1.62\\ (-38\%)\end{tabular}                     & \begin{tabular}[c]{@{}l@{}}382M vs 392M\\ (-2.6\%)\end{tabular} & \begin{tabular}[c]{@{}l@{}}382M vs 458M\\ (-17\%)\end{tabular} \\ \hline
%	\end{tabular}
%\end{table}

\begin{table}[]
	\caption{\textbf{\revision{Relative Metrics of VIP vs Default configuration as we increase the load factor (\textit{lf}) at $zipf=2$. The trends for low and high skew are similar.}}}\label{tab:roofline_lf_metrics}
	\revision{
	\begin{tabular}{|l|c|c|c|c|}
		\hline
		\textit{lf}& \textbf{\begin{tabular}[c]{@{}c@{}}Throughput \\ (fetch ops/s)\end{tabular}} & \textbf{\begin{tabular}[c]{@{}c@{}}Avg. Disp- \\ -lacement\end{tabular}} & \textbf{\begin{tabular}[c]{@{}c@{}}L3 \\Misses\end{tabular}}                                                   & \textbf{\begin{tabular}[c]{@{}c@{}}L1 \\ Misses\end{tabular}}                                                  \\ \hline
		0.5 & \begin{tabular}[c]{@{}c@{}}235M \textit{vs} \\ 188M\\ \textit{(+25\%)}\end{tabular}      & \begin{tabular}[c]{@{}c@{}}1.0 \textit{vs} 1.03\\ \textit{(-3\%)}\end{tabular}          & \begin{tabular}[c]{@{}c@{}}378M \textit{vs}\\ 385M\\ \textit{(-1.8\%)}\end{tabular}   & \begin{tabular}[c]{@{}c@{}}380M \textit{vs} \\ 412M\\ \textit{(-8\%)}\end{tabular}  \\ \hline
		1   & \begin{tabular}[c]{@{}c@{}}236M \textit{vs} \\ 134M\\ \textit{(+77\%)}\end{tabular}      & \begin{tabular}[c]{@{}c@{}}1.0 \textit{vs} 1.17\\ \textit{(-15\%)}\end{tabular}         & \begin{tabular}[c]{@{}c@{}}376M \textit{vs} \\ 387M\\ \textit{(-2.6\%)}\end{tabular} & \begin{tabular}[c]{@{}c@{}}380M \textit{vs} \\ 436M\\ \textit{(-13\%)}\end{tabular} \\ \hline
		1.5 & \begin{tabular}[c]{@{}c@{}}236M \textit{vs} \\ 90M\\ \textit{(+160\%)}\end{tabular}      & \begin{tabular}[c]{@{}c@{}}1.0 \textit{vs} 1.62\\ \textit{(-38\%)}\end{tabular}         & \begin{tabular}[c]{@{}c@{}}382M \textit{vs}\\ 392M\\ \textit{(-2.6\%)}\end{tabular}  & \begin{tabular}[c]{@{}c@{}}382M \textit{vs} \\ 458M\\ \textit{(-17\%)}\end{tabular} \\ \hline
	\end{tabular}}
\end{table}

	\section{Learning Popularity on-the-fly}\label{sec:learning}
	In this section, we highlight the challenges of learning in-the-loop (\S\ref{sec:learning_is_costly}), which motivated the lightweight mechanisms we built for VIP hashing. We describe how we learn, adapt, sense, and dynamically control the overhead on the fly (\S\ref{sec:vip_hashing}-3).

%In this section, we motivate our mechanisms for discuss our mechanisms to learn, adapt, sense, and dynamically switch-on/off learning. We show how we can learn the popularity distribution on the fly, and adapt to changing hot sets.

\subsection{Learning In-the-Loop is Costly}\label{sec:learning_is_costly}

\begin{figure}
	\centering
	\subfloat[Loss in performance when adding a 1-byte counter per key in the hash table. Both hash tables are identical (in Default configuration) except for the size of the entries (16 vs 17 bytes).]{
		\includegraphics[scale=0.35]{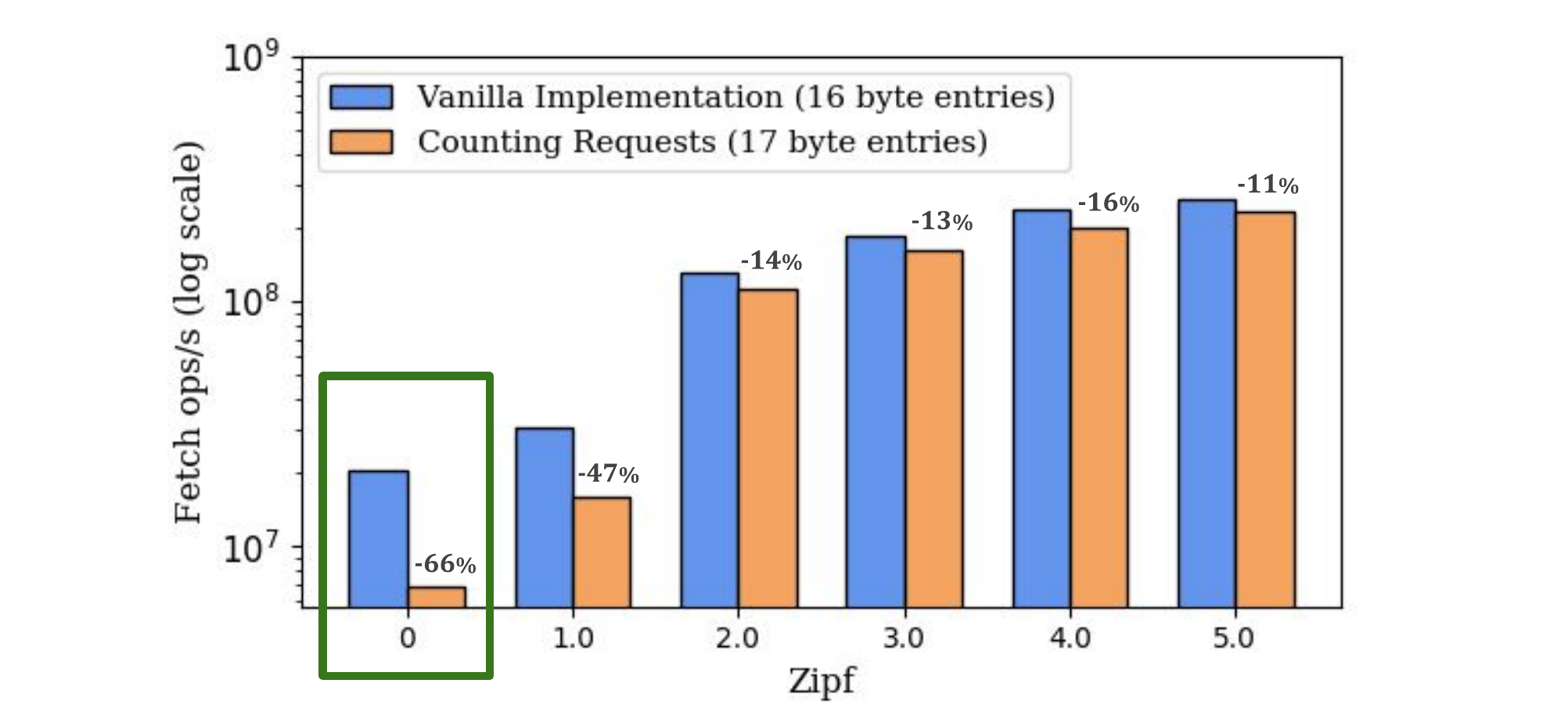}
		\label{fig:adding_counter_all_zipfs}
	}

	\subfloat[Relative metrics for zipf = 0. Instructions executed and cache misses increase after adding the 1-byte counter.]{
		\includegraphics[scale=0.35]{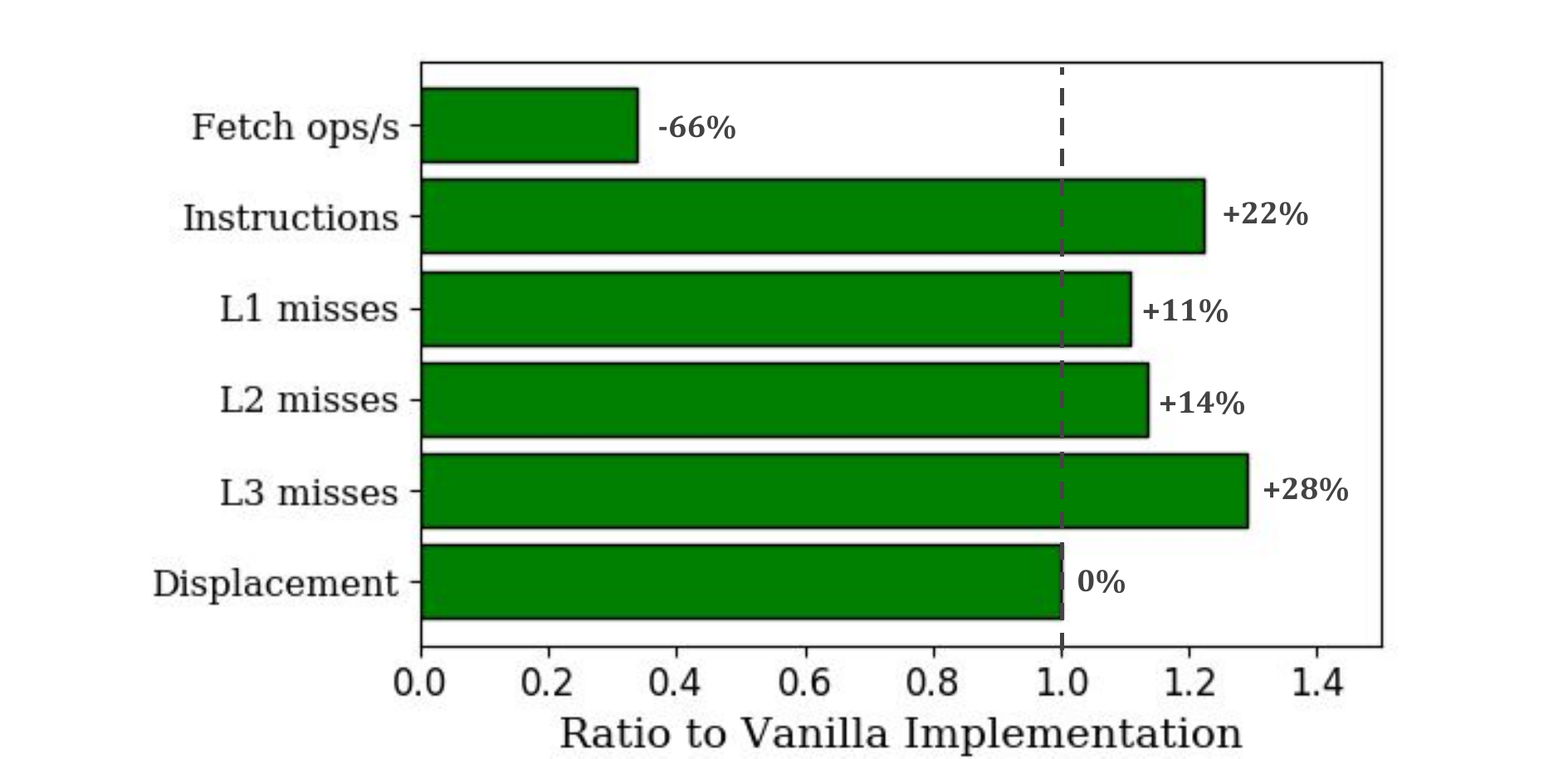}
		\label{fig:adding_counter_zipf0}
	}
	\caption{\textbf{The effect of adding a 1-byte requests counter per key in the hash table. 500M fetch operations are issued to a hash table with 1M keys at load factor $0.95$. Performance can take a significant hit \textendash\ we observe a 66\% loss in fetch operation throughput at $zipf=0$. This experiment demonstrates the challenges of learning in-the-loop with hash tables.}}
	\label{fig:adding_counter}
	\vspace{-3mm}
\end{figure}

Hash tables execute a tight loop of instructions \textendash\ compute the hash function, access keys in the bucket, and perform required operations to process the request. Adding any amount of additional computation or storage to this loop can have a significant impact on the performance. To demonstrate this behavior, we conduct a simple experiment of adding a 1-byte requests counter per key in the hash table, such that the entries become 17 bytes long (8 byte key and value, and 1 byte counter).

We use Wiscer to compare the performance of the vanilla implementation of hash table (16 byte entries) to the implementation with request counters (17 byte entries). We issue 500M fetch requests to a hash table with 1M entries (load factor $0.95=10^6/2^{20}$) for different levels of skew in the popularity distribution ($zipf=0$ to $5$ in steps of $1$). The remaining configuration options of Wiscer are set to the defaults (refer to Table~\ref{tab:wiscer_options}). Fig.~\ref{fig:adding_counter} shows the relative performance of the two hash table implementations at different levels of skew in the workload. There is a significant loss in throughput ranging from 11-66\% due to increase in cache misses and instructions executed.
%%We compare the performance metrics obtained from Wiscer (Fig.~\ref{fig:adding_counter}) for the two implementations.
%Hashing is a tight loop low latency operation. Counting is a basic requirement for learning.
%Show how the performance of fetch operation degrades when a simple 1-byte counter per entry is added. Show how cache misses increase, computation increases as well.

%Show how MRU would perform. This involves adding some computation. Cache misses should remain roughly the same. Its hard to predict how this will behave.

%For zipf values ranging from 0 to 5, different level of loss in performance. For zipf 0, 73\% loss in fetch throughput is observed. For higher levels of skew, different amounts of loss are observed. Hash table of 10M keys. Random seed is 0. 1B fetch operations are issued. Compare throughput to when there is no per-entry counter. Plot generated will be bar chart of relative metrics (normalized by baseline metrics). Increased cache misses and computation are the reasons for this loss. Hashing is extremely cache sensitive and instruction sensitive.

Counting requests is a fundamental requirement for learning the popularity distribution. However, this experiment shows that even adding a small amount of additional memory can hurt performance significantly in the extreme case. Thus, the challenge here is to work with a restricted ``budget" when learning in-the-loop, to balance the gains against the overhead of learning.

%Thus, we are constrained by a strict budget when learning in the looop.

%Any extra memory overhead is going to hurt performance.
%
%Any extra computation is going to hurt performance.

\subsection{VIP Hashing}\label{sec:vip_hashing}
From \S\ref{sec:learning_is_costly}, we know that using additional memory and computation can really hurt the performance of hash tables.
\revision{In this section, we describe how VIP hashing overcomes these challenges} by using lightweight mechanisms for learning and adapting to the popularity distribution (\S\ref{sec:vip_hashing_learning_adapting}), while controlling the overhead by {sensing} and {dynamically switching-on/off} learning as necessary (\S\ref{sec:vip_hashing_sensing_toggle}). We first give an overview of VIP hashing (\S\ref{sec:vip_hashing_overview}) followed by describing the mechanisms used in detail (\S\ref{sec:vip_hashing_learning_adapting}-3).
%From \S\ref{sec:learning_is_costly}, we learnt that any amount of computation and additional storage can really hurt performance in the worst case. VIP hashing uses lightweight mechanisms to learn, sense, adapt, and dynamically swich on/off learning to control the overhead. We first give an overview of VIP hashing, and then discuss each of the mechanisms individually in detail.

\begin{figure*}
	\centering
	\includegraphics[scale=0.41]{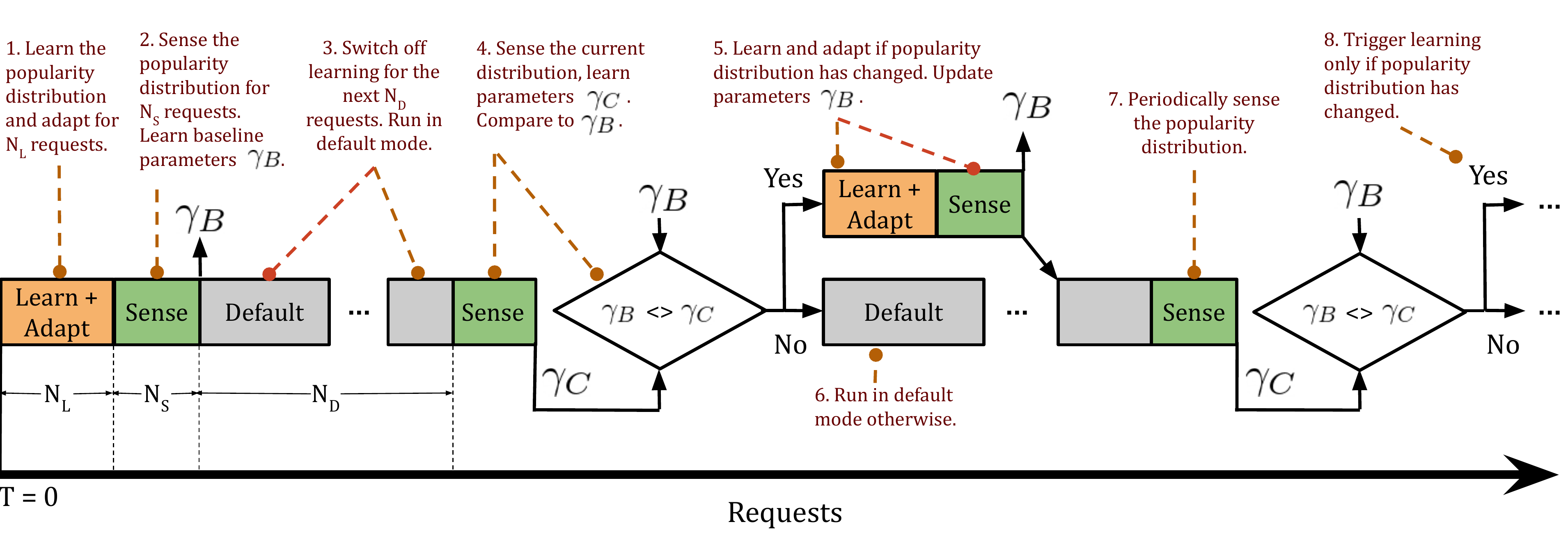}
	\caption{\textbf{Overview of VIP Hashing. \revision{At any time, the hash table is in one of the three modes \textendash\ \textit{learn+adapt}, \textit{sense}, or \textit{default}.} The amount of time spent on learning and adapting is limited since it is costly. The popularity distribution is sensed periodically to detect changes and trigger learning only when necessary.}}
	\label{fig:vip_overview}
\end{figure*}

\subsubsection{Overview}\label{sec:vip_hashing_overview}

Fig. \ref{fig:vip_overview} shows the VIP hashing method. At any given time, there are three possible {modes} that the hash table implementation can be in \textendash\ \textit{learn+adapt}, \textit{sense}, and \textit{default} (or vanilla). In the learn+adapt mode, the hash table learns the popularity distribution and rearranges keys to move closer to the VIP configuration. This mode is costly in terms of both computation and storage, and we control how much we run this mode by configuring the parameter N\textsubscript{L}. The learn+adapt mode is run at the start, and subsequent triggers of this mode happen only if the popularity distribution changes, which is determined during the sense mode.

The sense mode is triggered after the learn+adapt mode to measure some statistics ($\gamma$\textsubscript{B}) that characterize the popularity distribution.  These statistics require a total of 24 bytes of memory for the whole hash table (irrespective of the size) and a few additional arithmetic operations in the loop. Since the memory and computation footprint of this mode is low, it does not add much overhead to the execution. The sense mode is run for N\textsubscript{S} requests at a time, and is triggered periodically (every N\textsubscript{D} requests) to characterize the popularity distribution at the time ($\gamma$\textsubscript{C}). Comparing the statistics ($\gamma$\textsubscript{B} and $\gamma$\textsubscript{C}) helps determine if the popularity distribution has changed, and informs the decision of whether to switch on learning.
% The sense mode is run for N\textsubscript{S} requests (\S\ref{sec:vip_hashing_parameters}).

The default mode is the vanilla implementation of chained hashing (\S\ref{sec:background_ht}) with 16 byte entries. There is no additional overhead of storage or computation. This mode is run most of the time (N\textsubscript{D} $>$ N\textsubscript{L}, N\textsubscript{S}), so the performance is close to the vanilla implementation of hash table in the worst case.

In the following sections, we discuss the mechanisms we use for the learn+adapt (\S\ref{sec:vip_hashing_learning_adapting}) and sense (\S\ref{sec:vip_hashing_sensing_toggle}) modes. We discuss our choice of parameters (N\textsubscript{L}, N\textsubscript{S}, N\textsubscript{D}, etc.) in \S\ref{sec:vip_hashing_parameters}, that allow us to balance the performance gains against the overhead of learning.

%Fig.~\ref{fig:vip_overview} shows how the VIP hashing method progresses as requests are processed. The learn+adapt mode is run at the start for N\textsubscript{L} requests followed by sensing the popularity distribution for the next N\textsubscript{S} requests to obtain statistics $\gamma$\textsubscript{B}. The sense mode is triggered every N\textsubscript{D} requests to compute statistics $\gamma$\textsubscript{C}, which are compared to $\gamma$\textsubscript{B} to determine if the popularity distribution has diverged. Learning is triggered only if the popularity distribution has changed.

\subsubsection{Learning \& Adapting}\label{sec:vip_hashing_learning_adapting}

Algorithm~\ref{alg:learning} describes how we learn the popularity distribution and adapt to the skew on the fly. The popularity of a key is estimated as the proportion of requests made to the key (\S\ref{sec:background_prob}). Thus, learning the popularity distribution requires counting requests, which we know is challenging from the discussion in \S\ref{sec:learning_is_costly}.

\setlength{\textfloatsep}{10pt}
%\vspace{1pt}
\begin{algorithm}[t]
	\caption{Learning and Adapting on-the-fly}\label{alg:learning}
	\begin{algorithmic}[1]
		\Procedure{FetchAdaptive}{requests}
		\State \textit{ht} $\gets$ getHashTable()
		\State /* Requests are counted in a separate data structure*/
		\State \textit{req\_cnt\_ht} $\gets$ getRequestsCountingHashTable()
		\For{\textit{r} \textbf{in} \textit{requests}}
		\State \textit{hash} $\gets$ murmurHash(\textit{r.key})
		\State \textit{ht\_entry} $\gets ht[hash]$
		\State \textit{req\_entry} $\gets req\_cnt\_ht[hash]$
		\State /* Keep track of entry with minimum requests */
		\State \textit{min\_req\_ht\_entry} $=\ ht\_entry$
		\State \textit{min\_req\_entry} $=\ req\_entry$
		\While{$ht\_entry$ \textbf{and} $ht\_entry.key\ \neq\ r.key$}
		\If {$req\_entry.count < min\_req\_entry.count$}
		\State $min\_req\_ht\_entry\ =\ ht\_entry$
		\State $min\_req\_entry\ =\ req\_entry$
		\EndIf
		\State $ht\_entry\ =\ ht\_entry.$next()
		\State $req\_entry\ =\ req\_entry.$next()
		\EndWhile
		\If {$ht\_entry == null$}
		\State $r.found = $ false
		\State \textbf{continue}
		\EndIf
		\State $r.found = $ true	
		\State $r.value\ =\ ht\_entry.value$
		\State $req\_entry.count = req\_entry.count+1$
		\If {$req\_entry.count > min\_req\_entry.count$}
		\State /* Swap this entry with the min requests entry */
		\State swap(\textit{ht\_entry, min\_req\_ht\_entry})
		\State swap(\textit{req\_entry}, \textit{min\_req\_entry})
		\EndIf
		\EndFor
		%		\State \Return \textit{requests}
		\State \revision{/* Reclaim cache space by clearing $req\_cnt\_ht$ */}
		\State \revision{clearCache($req\_cnt\_ht$)}
		\EndProcedure
	\end{algorithmic}
\vspace{-2pt}
\end{algorithm}

\revision{To overcome the challenge of counting requests in-the-loop, we perform two optimizations}. First, we count requests in a separate data structure that mimics the hash table in arrangement (for every entry in the hash table, there is a corresponding entry in the request counting hash table). Although this temporarily requires more memory \revision{(about 50-60\% increase in memory usage depending on the load factor)} than maintaining a counter per key in the hash table, the cost is incurred only during the learn+adapt mode. \revision{Second, at the end of the learn+adapt mode, we clear the request counting hash table ($req\_cnt\_ht$) from the cache by issuing cache flush instructions (\texttt{\_mm\_clflushopt} on Intel CPUs~\cite{intelintrinsics}), thus restricting the cache pollution caused by the additional data structure to learn+adapt mode.}

%To overcome this challenge of counting, we maintain a separate data structure to count requests that mimics the hash table in arrangement (for every entry in the hash table, there is a corresponding entry in the requests counting hash table). Although this requires more memory than maintaining a counter per key in the hash table, the cost is incurred only during the learn+adapt mode. Once the learn+adapt mode is switched off, the cache space occupied by the requests counting hash table ($req\_cnt\_ht$ in Algorithm~\ref{alg:learning}) can be reclaimed without affecting performance in the subsequent sense and default modes.

%It is important to note that we maintain the accesses data structure separate from the hash table while replicating the structure. This is done so that once we switch off the learning mode, the cache space occupied by the separate accesses data structure can be reclaimed by the original hash table. Otherwise we will see a loss in performance similar to Fig. \ref{fig:adding_counter}.

To attain the VIP configuration, we need to sort the keys in descending order of popularity in the bucket chains. Given that the proportion of requests made to a key is an estimate of popularity, we use Algorithm~\ref{alg:learning} to stochastically sort the keys in descending order of requests received on the fly. When performing a fetch operation, we keep track of the entry with minimum requests ($min\_req\_ht\_entry$) encountered in the path to the entry being fetched. If the entry being fetched has received more requests, then it is swapped with the $min\_req\_ht\_entry$ and it moves forward in the chain. We propose the following theorem:

\begin{theorem}
	%	Let $f$ be continuous on $[a,b]$.  If $G$ is
	%	an antiderivative for $f$ on $[a,b]$, then
	%	\begin{displaymath}\int^b_af(t)dt = G(b) - G(a).\end{displaymath}
	Let there be a bucket chain with $n$ keys $K_1,$ $K_2\ ...\ K_n$ which have popularity $p_1 > p_{2}... > p_n > 0$. Let the keys be in a random order in the chain. Then, by applying Algorithm~\ref{alg:learning}, the keys will converge to the sorted order of popularity as number of fetch requests $N \rightarrow\infty$.
\vspace{-2pt}
\end{theorem}

We formally prove this theorem in Appendix~\ref{sec:appendix_a}. There are two properties of Algorithm~\ref{alg:learning} that are worth noting. First, the VIPs move to the front quickly, as they can skip over multiple entries in the chain in a single fetch request. This algorithm is, in essence, similar to selection sort as we are moving the entry with minimum requests to the end of the (sub-)chain being accessed. An alternative would be to compare only adjacent keys (bubble sort), which empirically requires more requests for a VIP to move to the front.

Second, the cost of swapping is amortized, as there is at most one swap performed per fetch operation. This approach is faster compared to performing a full sort on every request, or sorting at the end after counting requests for some time (we will have to access all the buckets in order
to perform a full sort; this will incur cache misses and also
pollute the cache).

%Essentially, this algorithm does stochastic sorting. In practice, the VIPs get more accesses and move to the front of the list pretty fast using this algorithm, as it is able to skip multiple intermediate entries at once. Also, there is only one swap happening every fetch request, so the cost of sorting is amortized.

\subsubsection{\revision{Sensing \& Dynamically Switch-on/off Learning}}\label{sec:vip_hashing_sensing_toggle}

Algorithm~\ref{alg:sensing} describes how we sense some key statistics of the popularity distribution, which enable us to dynamically switch-on learning only when the distribution has changed (Algorithm~\ref{alg:dynamic_toggle}). While there are multiple ways to quantify the difference between two probability mass functions (\textit{pmf}s) ~\cite{kldivergence,ztest,hellingers}, we choose a lightweight statistic to compare distributions \textendash\ \textit{average displacement}. In \S\ref{sec:roofline_results}, we saw that displacement encountered indicates the ``goodness" of the hash table configuration. Every popularity distribution imposes a pmf over the displacement encountered on a request, which is a derived random variable. Formally stated:
\newtheorem{axiom}{Axiom}
\begin{axiom} Let $K_1,\ K_2,\ ...,\ K_N$ be $N$ keys in the hash table with popularity $p_1,\ p_2,\ ..., p_N$ $\left(\sum p_i=1\right)$ at displacement $d_1,\ d_2,\ ...,\ d_N$ $\left(d_i \leq N\right)$. Let D be the random variable of the displacement encountered on a successful fetch request. Then,
\vspace{-2mm}
\begin{displaymath}
	P(D=d) = \sum_{i=1}^{N}p_i\cdot\mathbbm{1}_{d_i=d}
\end{displaymath}
\end{axiom}
\vspace{-1mm}
i.e, the probability that displacement $d$ is encountered on a successful fetch request is the probability that any of the keys with displacement $d$ were fetched. The average displacement is calculated as
\vspace{-1mm}
\begin{displaymath}
\mu_{D} = E[D] = \sum\limits_{i=1}^{N}i\cdot P(D=i)
\end{displaymath}
%\vspace{-1mm}
We make the following observation:
\begin{axiom}
	 The VIP configuration minimizes E[D] over all possible arrangements of keys in the hash table for a fixed load factor, popularity distribution, and hash function.
\end{axiom}
The VIP configuration orders keys by popularity, thus giving more ``weight'' to lower values of $D$ which minimizes the average displacement. It is straightforward to see that for a given hash table configuration, two popularity distributions with different average displacement will not be identical (although the opposite is not true). Thus, a change in average displacement reflects a shift in the popularity distribution.

%Axiom 2 change in avg displacement means that popularity distribution has shifted. Other way round is not true, but it is ok.
%
%Basically, avg displacement is a measure of the goodness of the hash table arrangement. VIP configuration will have the minimum avg displacement.

\begin{algorithm}[t]
	\caption{Sensing}\label{alg:sensing}
	\begin{algorithmic}[1]
		\Procedure{FetchSensing}{requests}
		\State \textit{ht} $\gets$ getHashTable()
		\State /* Metrics to track */
		\State \textit{disp} $\gets 0$ \algorithmiccomment{{cumulative displacement}}
		\State \textit{disp\_sq} $\gets 0$ \algorithmiccomment{{cumulative disp. square}}
		\State \textit{count} $\gets 0$ \algorithmiccomment{number of requests}
		\State \textit{c} $= 0.95$ \algorithmiccomment{confidence level of the interval}
		\For {\textit{r} \textbf{in} \textit{requests}}
		\State \textit{hash} $\gets$ murmurHash(\textit{r.key})
		\State \textit{ht\_entry} $\gets ht[hash]$ 
		\State \textit{d} $\gets 1$
		\While {\textit{ht\_entry} \textbf{and} \textit{ht\_entry$\rightarrow key \neq r.key$}}
		\State \textit{ht\_entry} $=\ ht\_entry.$next()
		\State $d = d + 1$
		\EndWhile
		\If {$ht\_entry == null$}
		\State $r.found =$ false
		\State \textbf{continue}
		\EndIf
		\State $r.found =$ true
		\State $r.value = ht\_entry.value$
		\State $count = count + 1$
		\State $disp = disp + d$
		\State $disp\_sq = disp\_sq + d\times d$
		\EndFor
		\State /* Estimating mean \textit{u}, variance \textit{v}, and C.I. width \textit{w}*/
		\State $u = disp/ count$
		\State $v = disp\_sq/(count-1) - disp^2/(count*(count-1))$
		\State $w = \sqrt{-2.v.log(1-c)/count}$ \algorithmiccomment{Gaussian tail bound}
		\State $\gamma = (u, w)$
		\State \Return $\gamma$
		\EndProcedure
	\end{algorithmic}
\end{algorithm}

The parameters we learn from sensing are $\gamma =(\hat{\mu}_D, \hat{w}_D) = (u,w)$ (Algorithm~\ref{alg:sensing}), where $\hat{\mu}_D$ is the estimated average displacement, and $\hat{w}_D$ is the width of the confidence interval around $\hat{\mu}_D$ obtained using Gaussian tail bounds (\S\ref{sec:background_prob}). We estimate the average displacement as
\begin{algorithm}
	\caption{\revision{Dynamically Switch-on/off Learning}}\label{alg:dynamic_toggle}
	\begin{algorithmic}[1]
		\Procedure{HasDistributionChanged}{$\gamma_B, \gamma_C$}
		\State $(u_B, w_B) = \gamma_B$
		\State $(u_C, w_C) = \gamma_C$
		%		\State $diff = |u_B - u_C5|$
		\If {$|u_B - u_C| > (w_B + w_C)$}
		\State \Return true
		\Else
		\State \Return false
		\EndIf
		\EndProcedure
	\end{algorithmic}
\end{algorithm}
\begin{displaymath}
	\hat{\mu}_{D} = \frac{\sum\limits_{i=1}^{N_S}D_i}{N_S}
\end{displaymath}
which is the sample mean\footnote{Note that instead of sampling, we could also use the request counting data structure ($req\_cnt\_ht$ in \S\ref{sec:vip_hashing_learning_adapting}). However, this would incur cache misses and also pollute the cache affecting performance (\S\ref{sec:learning_is_costly}).} of displacement encountered $D_i$ ($1\leq i\leq N_S$) over $N_S$ fetch requests in the sense mode. Similarly, we also estimate sample variance $\hat{\sigma}_D^2$ (\S\ref{sec:background_prob}).

We further characterize the pmf by building a confidence interval using Gaussian tail bounds (\S\ref{sec:background_prob}). The width ($
\hat{w}_D$) of the interval at confidence level $c$ ($c=0.95$ in our experiments) is calculated as
\begin{displaymath}
	\hat{w}_D = \sqrt{\frac{-2\cdot\hat{\sigma}_D^2\cdot(1-c)}{N_S}}
\end{displaymath}
Note that $\hat{\sigma}_D$ is estimated variance from a sample of $N_S$ observations, and $(\hat{\mu}_D-\mu_D)$ only approximately Gaussian according to CLT (\S\ref{sec:background_prob}). Thus, the width $\hat{w}_D$ obtained by applying Gaussian tail bounds is a heuristic.

We switch-on learning (Algorithm~\ref{alg:dynamic_toggle}) only if we detect a significant change in the average displacement. Given two sets of parameters $\gamma_B=(u_B,w_B)$ and $\gamma_C=(u_C,w_C)$ where $u_B$ and $u_C$ are estimated means, we check if the confidence intervals are disjoint. If so, then heuristically with a probability $c^2=(0.95)^2=0.9$, we can be sure that the real means are not equal and the distributions have diverged. Thus, we detect changes in popularity distribution in a non-intrusive manner by computing lightweight statistics.

%We characterize the popularity distribution by building a confidence interval using the gaussian tail bound. We estimate mean and variance using sampling. Using the gaussian tail bound, we obtain the width of the confidence interval with 95\% confidence. This is a heuristic value.
%
%To compare the distribution, we check if the confidence intervals overlap. Specifically, some math here.
%
%If the confidence intervals do not overlap, we can be 90\% confident that avg displacement value has changed, and by Axion 2, the probability distribution has changed as well.

\subsection{\revision{Parameters}}\label{sec:vip_hashing_parameters}

\revision{The parameters N\textsubscript{L}, N\textsubscript{S}, and N\textsubscript{D} determine how long the hash table runs in learn+adapt, sense, and default modes respectively. Our goal is to choose these parameters such that the gains of learning are balanced against the overhead.}

\revision{Our choice of parameters is general, made using theoretical and empirical evidence that is independent of the popularity distribution. Thus, our techniques (\S\ref{sec:vip_hashing}) apply to any distribution with skew irrespective of its specific properties. Note that it is possible to further tune the parameters and the techniques with additional knowledge such as total number of requests, patterns in the workloads, family of distribution, etc.}

%	 We use theoretical and empirical evidence inform the choice of these parameters. %\revision{Note that this is a general choice of parameters made while planning for the worst case. With additional knowledge (and/or assumptions) about the nature of the distribution, amount of skew, known patterns in popularity of keys, total number of requests to be issued, the request distribution, etc. the choice of these parameters can be further improved.}}

\subsubsection{\revision{Allocating the budget for learning $\textendash$ N\textsubscript{L} \textit{vs} N\textsubscript{D}}}\label{sec:learning_budget}

Learning in-the-loop is costly \textendash\ in our experiments, we find that the learn+adapt mode can be as much as $4x$ slower than the vanilla implementation in the worst case under no skew (we tested different hash table sizes from 1M to 100M keys). If a total of $(N_L+N_D)$ requests are issued to VIP hashing, the loss in throughput due to learning would be:
\begin{displaymath}
	1 - \frac{T_{vanilla}}{T_{vip}} \leq \Big(1-\frac{N_D.t+N_L.t}{N_D.t+N_L.4.t}\Big)
\end{displaymath}

assuming that the vanilla implementation takes time $t$ on an average to process each request. We cap the overhead of learning to at most 5\% by choosing $N\textsubscript{D}=60\cdot N\textsubscript{L}$ in our experiments (i.e, learn+adapt mode is run for at most $\sfrac{1}{61}$ of the total requests). More generally, the cap on overhead is ($1-\sfrac{61}{(60+k)}$), where $k$ depends on the experimental configuration ($k=4$ on our hardware). Thus, fixing a budget for $\sfrac{\text{N\textsubscript{L}}}{\text{N\textsubscript{D}}}$ limits the overhead of learning in the worst case.

\subsubsection{\revision{Choosing $N\textsubscript{L}$ $\textendash$ how much to learn?}}

The learn+adapt mode is run for $N_L$ requests at a time. \revision{Our goal is to capture the popularity distribution as much as possible while learning for a finite number of requests.} From previous work~\cite{statisticalcomplexity}, we know that it takes $\Theta(N)$ i.i.d. samples to learn a probability mass function over $N$ items (with error $\epsilon=1$ in KL divergence compared to the true pmf). \revision{When the cardinality of the hash table is not known/can vary, we choose $N_L = 1.5\cdot (htsize)$, i.e, 1.5 times the number of buckets in the hash table. Since we maintain a load factor of at most $1.5$ at all times, the number of keys in the hash table $N\leq 1.5\cdot htsize$, which satisfies our requirements.}

\subsubsection{\revision{Parameters for sensing $\textendash$ $N\textsubscript{s}$ and \textit{c}}}

We sense the distribution for N\textsubscript{S} requests at a time to estimate the average displacement $\hat{\mu}_{D}$ and build an interval with confidence $c$. Since the load factor is low and the longest chain length is likely to be low as well (except in pathological cases where many keys are hashed to the same bucket), we have found that choosing $N_S$ to be a large number (1000) has been sufficient in our experiments. We build a $c=95\%$ confidence interval that heuristically gives us a probability of $c^2=(0.95)^2=0.9$ when we detect a shift in popularity. By increasing (decreasing) the confidence level, we can be less (more) sensitive to changes in popularity.

	\section{\revision{Applications}}\label{sec:evaluation}
	%\revision{In this section, we evaluate VIP hashing for two common uses of hash tables \textendash\ primary key-foreign key hash joins (\S\ref{sec:pk_fk_join}), and point lookups in  (\S\ref{sec:point_lookups}).}

\subsection{\revision{PK-FK Hash Joins}}\label{sec:pk_fk_join}

\revision{Hash tables are frequently used in database systems for processing join queries. In this section, we describe how VIP hashing can improve the performance of primary key-foreign key (PK-FK) hash joins in the presence of skew.}

\subsubsection{\revision{Experimental Setup}}
\revision{Motivated by past research~\cite{alonso,spyros,hjpaper}, we consider the canonical PK-FK join query on tables $R$ and $S$ ($|R| \leq |S|$) with 8-byte integer attributes (16-byte tuples). Skew can arise in PK-FK relations~\cite{spyros,alonso} when some keys occur more frequently than others in the outer relation $S$. We use Wiscer to instantiate $R$ and $S$ using the sequential key pattern for primary keys in $R$,  and varying the level of skew in the outer relation $S$ from uniform ($zipf=0$) to high ($zipf=3$) for 10 distinct random seeds. We compare the performance of the canonical hash join algorithm~\cite{alonso,hjpaper} implemented using the default and VIP hash tables, while materializing pointers to output tuple pairs. We assume that the tuples in $S$ are i.i.d, i.e, the popularity distribution is static. We explore effects of dynamic popularity distribution in \S\ref{sec:point_lookups}.}

\subsubsection{\revision{Default vs VIP Hash Join}}

\begin{figure}
	\centering
	\includegraphics[scale=0.46]{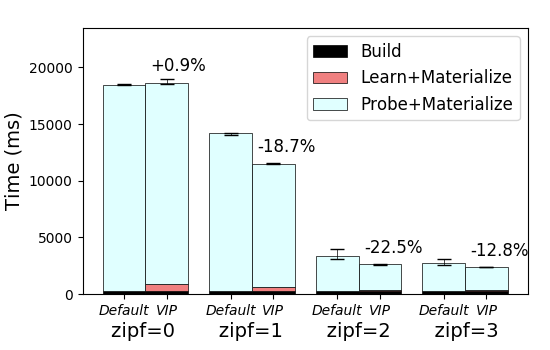}
	\caption{\revision{\textbf{Performance of PK-FK canonical hash join on tables $R$ and $S$ ($|R|:|S| = 1:16$), when using the default and VIP hash table implementations. For medium skew, we observe a 22.5\% reduction in median (over 10 random seeds) total execution time.}}}
	\label{fig:eval_hashjoin}
	\vspace{4mm}
\end{figure}

\begin{table}[]
	\caption{\revision{\textbf{Relative metrics for default and VIP hash join at $zipf=2$, $|R|:|S| = 1:16$.}}}
	\label{tab:hashjoin_metrics}
	\centering
	\revision{
		\begin{tabular}{l|c|c|c}
			\hline
			\textit{Metric}            & \textit{Default} & \textit{VIP} & \textit{Diff} \\ \hline
			\textit{Time}              & 3.4s             & 2.6s         & -22.5\%       \\ 
			\textit{Avg. Displacement} & 1.23             & 1.0003       & -18.7\%       \\ 
			\textit{L3 Misses}         & 75.5M            & 75.3M        & -0.3\%        \\ 
			\textit{L2 Misses}         & 127.9M           & 124.6M       & -2.6\%        \\ 
			\textit{L1 Misses}         & 161.2M           & 155.7M       & -3.4\%        \\ 
			\textit{Instructions}      & 8.5B             & 8.2B         & -3.5\%        \\ \hline
		\end{tabular}
	}
\end{table}

\revision{Fig.~\ref{fig:eval_hashjoin} shows the relative execution time of the default vs VIP hash join implementations. The cardinalities of $R$ and $S$ are 12M and 192M respectively ($|R|:|S| = 1:16$)~\cite{spyros,alonso}, and the load factor is $1.4$ ($= 12\cdot 10^6/2^{23}$). For medium skew in the outer relation, the average displacement encountered by the default hash join implementation is $1.23$ (Table~\ref{tab:hashjoin_metrics})\footnote{\revision{Note that the average displacement is low for the default configuration in this case, since the keys are sequential. Holding the load factor constant, randomly generated keys result in a median (over 10 random seeds) average displacement of 1.48.}}.}

\revision{For the case of canonical hash join query, the learning budget of the VIP hash table implementation can be calculated in advance while maintaining $N_L:N_D = 1:60$ (\S\ref{sec:vip_hashing_parameters}) since we almost always know the cardinalities of the relations from system catalogs. Learning is triggered at the beginning of the probe phase with a budget of $N_L = min(|R|,\  \frac{|S|}{61})$ $= \frac{16\cdot|R|}{61} = 0.26\cdot|R|$ lookups from the outer relation. Learning takes about 3\% of the total execution time, ranging from 70-600ms depending on the level of skew. Note that the average displacement of the VIP hash join implementation is very close to 1 (Table~\ref{tab:hashjoin_metrics}) indicating that the learning mechanism efficiently captures the popularity distribution, and reduces cache misses and instructions executed.}

\revision{To show the impact of varying the learning budget, we repeated the experiment for lower and higher cardinality ratios. For a ratio of $1:4$, we have a learning budget of $\frac{4\cdot|R|}{61}=0.07\cdot|R|$ requests and the overall reduction in execution time is 18.6\%. On the other hand, a cardinality ratio of $1:64$ allows a learning budget of $|R|$ $=min(|R|,\frac{64\cdot|R|}{61})$ and results in 25.8\% reduction in execution time. Thus, the available learning budget impacts the gain in performance.}

% while adjusting the learning budget to $(0.07\cdot|R|)$ and $|R|$ respectively. Given that the learning budget is lower for the former case, we obtain a reduction in total execution time of 18.6\%. On the other hand, we obtain a reduction of 25.8\% in execution time for the 

\subsubsection{\shepherd{Application to Skewed TPC-H}}

\begin{figure}
	\centering
	\includegraphics[scale=0.18]{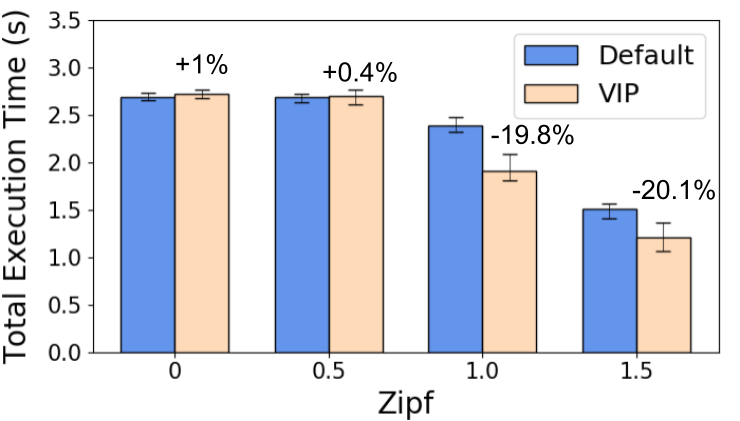}
	\vspace{-2mm}
	\caption{\shepherd{\textbf{Execution time of TPC-H query 9 (scale factor = 1) on DuckDB. VIP hashing speeds up PK-FK hash join probes, and results in 20\% reduction in median (over 10 random seeds) end-to-end query execution time at $zipf=1$ and $zipf=1.5$.}}}
	\label{fig:tpch}
	\vspace{2mm}	
\end{figure}

\shepherd{We focus our attention on TPC-H query 9~\cite{tpch}, which is the most expensive TPC-H query involving multiple PK-FK join operations. We implemented VIP hashing in DuckDB~\cite{duckdb}, an in-memory vectorized DBMS, to speed up PK-FK hash joins in single-threaded mode. Fig.~\ref{fig:tpch} shows the median execution time of VIP hash join relative to the default, tested on skewed TPC-H data (Appendix~\ref{sec:appendix_b}) at varying levels of skew (from $zipf=0$ to $zipf=1.5$) for 10 different random seeds. VIP hash join reduces the end-to-end query execution time by 20\% at $zipf=1$ and $zipf=1.5$, while the increase in execution time at lower skew is negligible. The remaining TPC-H queries spend $<$ 1\% of the total execution time in skewed PK-FK hash joins, and consequently the impact of VIP hashing is negligible.}

%This query spends majority of the time in PK-FK hash joins \textendash\ 45\% of the execution time is spent on probing the hash table for PK-FK hash join at zipf=1.5.
%We evaluate VIP hashing for skewed TPC-H data~\cite{} on DuckDB~\cite{}, an in-memory DBMS with a fully vectorized query execution engine, as the platform. We test VIP hashing on skewed TPC-H data~\cite{} Fig.~\ref{fig:tpch}. We implemented VIP hashing in DuckDB~\cite{}, an in-memory DBMS with a fully vectorized query execution engine.

\subsection{\revision{Point Queries}}\label{sec:point_lookups}

% Our goal is to explore if VIP hashing is able to adapt to the skew in popularity, while minimizing the overhead of learning.

\revision{Another common use of hash tables is for in-memory indexing in database systems~\cite{mysql,mariadbhash} and in key-value stores~\cite{redis,memcached} for processing point queries. In this section, we evaluate VIP hashing against a range of workloads generated using Wiscer that highlight the robustness of our techniques for learning in-the-loop under different conditions. In all the experiments, we assume no prior knowledge of the characteristics of the request distribution.} The first two workloads (\S\ref{sec:eval_static_popularity}-\S\ref{sec:popularity_churn}) involve fetch operations, and the last two (\S\ref{sec:eval_steady_state}-\S\ref{sec:eval_read_mostly}) perform insert and delete operations.

We run these workloads on a hash table with 1M entries \revision{(load factor 0.95 $= 10^6/2^{20}$)} in the Default configuration at the start. \shepherd{Each of these workloads issue 500M operations to the hash table at varying levels of skew ranging from uniform ($zipf=0$) to medium skew ($zipf=2$).} The remaining configuration options of Wiscer are set to the defaults (refer to Table~\ref{tab:wiscer_options}). We compare the performance of VIP hashing to the default hash table in Fig.~\ref{fig:vip_results_static}-\ref{fig:vip_results_read_mostly}.

% Somewhere I have to mention that we are only concerned about fetch requests. Inserts and deletes are one time operations, so they don't matter.

% Update operations are pretty much handled like fetch operations in hash tables. So we only consider fetch operations in our analysis.

%\subsection{Worst case - Uniform Popularity}\label{sec:uniform_popularity}
%
%Uniform distribution (zipf = 0) read only workload. Since there is no popularity skew, there can't be a change in popularity either.

\subsubsection{Static Popularity}\label{sec:eval_static_popularity}

In this workload, the popularity of keys in the hash table remains the same throughout the experiment. \shepherd{We run 500M fetch operations at four levels of skew from $zipf=0$ to $zipf=2$ in steps of $0.5$.} For the case of uniform popularity distribution ($zipf=0$), the overall loss in throughput is \revision{2\%} (Fig.~\ref{fig:static_popularity_zipf0}) which is within our allocated budget of 5\% (\S\ref{sec:learning_budget}), whereas for low skew ($zipf=1$), we obtain a net gain of \revision{22\%} (Fig.~\ref{fig:static_popularity_zipf1}). \shepherd{The performance gain is higher at medium levels of skew \textendash\ the gain in throughput at $zipf=1.5$ and $zipf=2$ is 77\% (Fig.~\ref{fig:static_popularity_zipf1_5}) and 116\% (Fig.~\ref{fig:static_popularity_zipf2}) respectively.} Since the popularity distribution is static, the learn+adapt mode is triggered only at the start of the experiment \revision{for $1.5\cdot htsize$ requests} in all cases. The periodic runs of the sense mode do not detect a change in popularity and the learn+adapt mode is not triggered again. Thus, learning is run only when necessary, and the overhead of VIP hashing is minimized

\begin{figure*}
	\centering
	\subfloat[Static popularity (\S\ref{sec:eval_static_popularity}) with $zipf=0$ (uniform distribution). Since there is no skew in popularity, no performance gain can be obtained from VIP hashing. Learning adds overhead to VIP hashing \revision{(4\textit{x} slower)}, and is only triggered at the start for \revision{$(1.5\cdot 2^{20})$ requests (0.3s)}. Subsequent sensing of the popularity distribution does not detect any change, and learning is not triggered. \revision{Total loss in throughput is 1.9\%, which is within our allocated budget.}]{
		\includegraphics[scale=0.46]{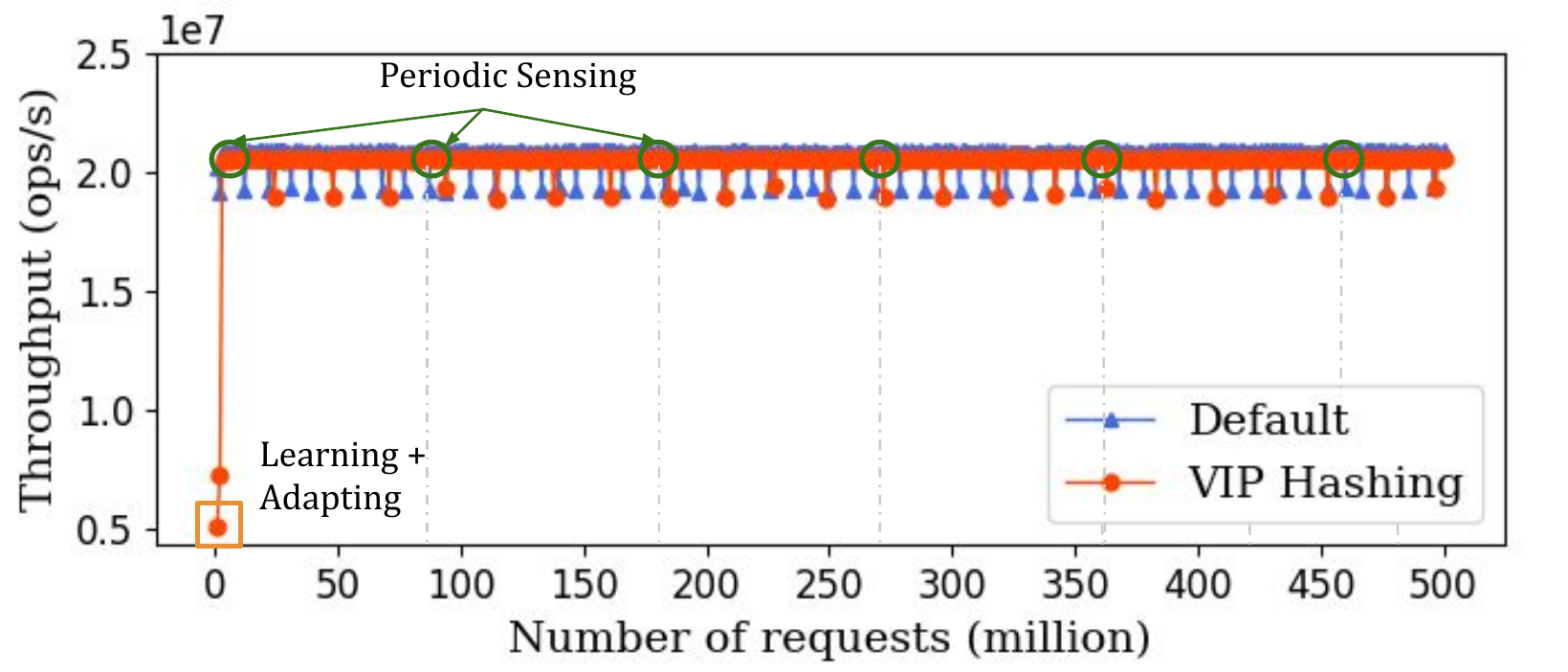}
		\label{fig:static_popularity_zipf0}
	}
	\hspace{4mm}
	\subfloat[Static popularity (\S\ref{sec:eval_static_popularity}) with $zipf=1$ (low skew). Learning is only triggered at the start and is 3x slower than the default \revision{(0.13s vs 0.05s respectively)}. Sensing does not detect any changes to the popularity distribution, so learning is not triggered again. The overhead of learning is offset by the gain in performance from the VIP configuration. \revision{We observe an overall increase in throughput of 21.8\%.}]{
		\includegraphics[scale=0.46]{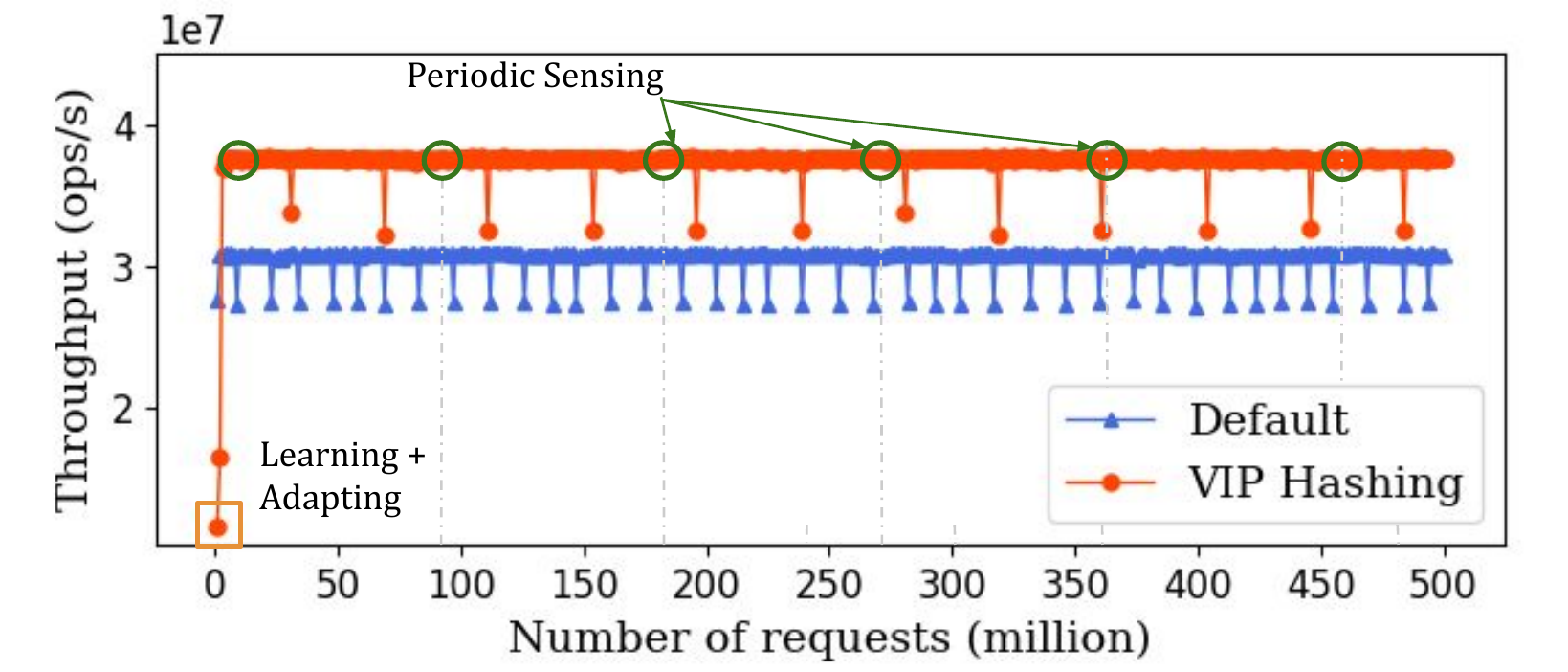}
		\label{fig:static_popularity_zipf1}
	}
	\vspace{3mm}
	\subfloat[\shepherd{Static popularity (\S\ref{sec:eval_static_popularity}) with $zipf=1.5$ (medium skew). Learning is only triggered at the start for 0.03s, and subsequent triggers of sense mode do not detect any changes to popularity. A net gain of 76.5\% in throughput is observed.}]{
		\includegraphics[scale=0.46]{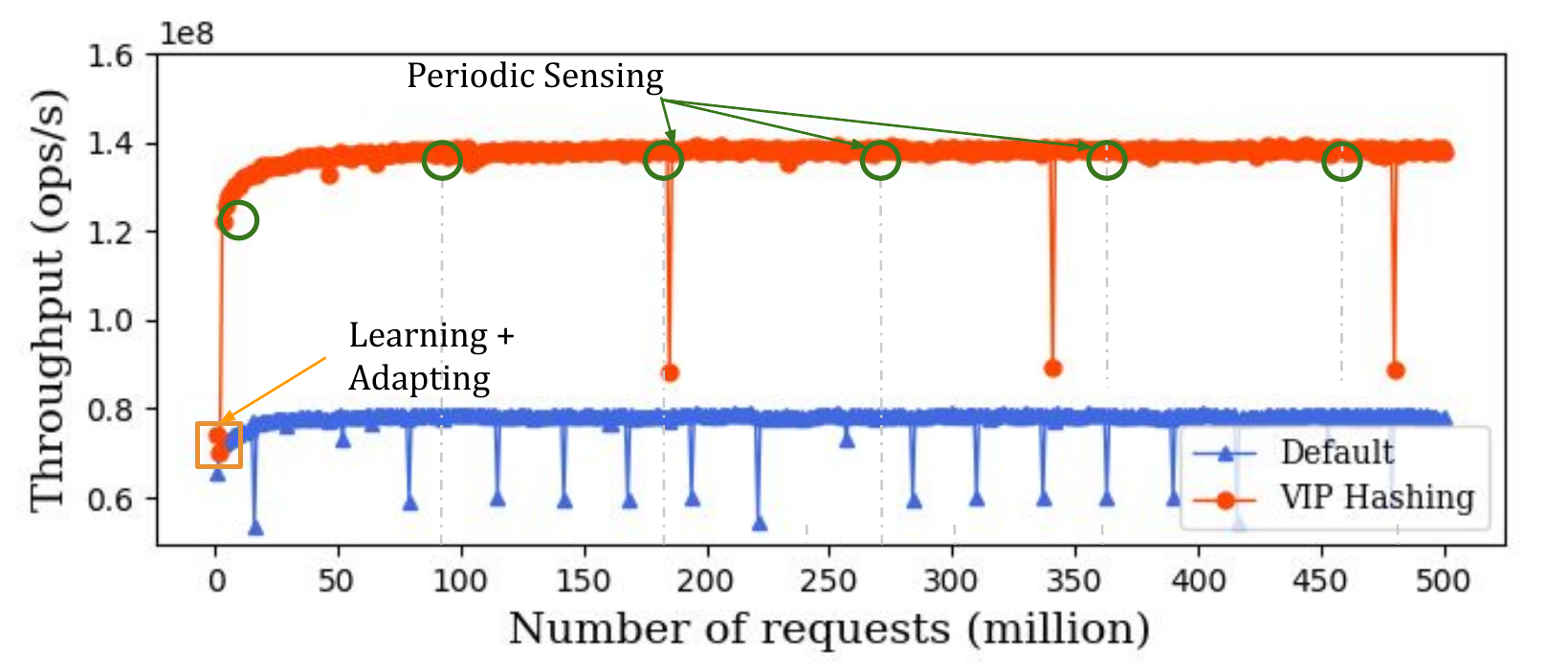}
		\label{fig:static_popularity_zipf1_5}
	}
	\hspace{4mm}
	\subfloat[\shepherd{Static popularity (\S\ref{sec:eval_static_popularity}) with $zipf=2$ (medium skew). Learning is only triggered at the start for 0.02s, and subsequent triggers of sense mode do not detect any changes to popularity. A net gain of 116.1\% in throughput is observed.}]{
		\includegraphics[scale=0.46]{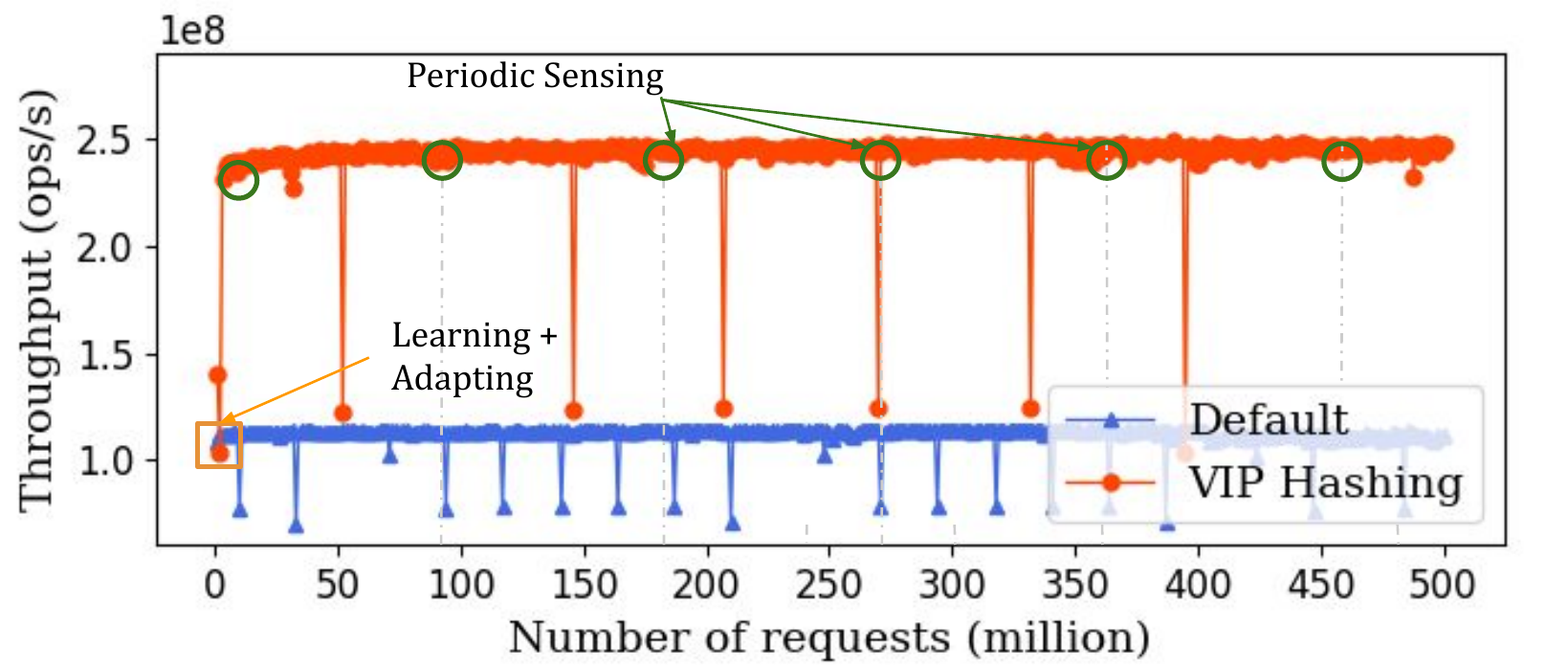}
		\label{fig:static_popularity_zipf2}
	}
	\caption{\textbf{\shepherd{Performance of VIP hashing under static popularity distribution at increasing levels of skew ranging from $zipf=0$ (uniform distribution) to $zipf=2$ (medium skew). Fetch requests are issued to a hash table with 1M entries at load factor 0.95 with keys in a random order initially. Learning is only triggered at the start\textsuperscript{2}, and the performance gain ranges from 22\% to 116\% depending on the level of skew. The overhead of VIP hashing in the case of uniform distribution ($zipf=0$) is 2\%, which is within our allocated budget.}}}
	\label{fig:vip_results_static}
\end{figure*}
%\footnote{The triggers of sense mode and learn+adapt mode have been marked using green circles and orange squares respectively. The unmarked periodic dips in throughput for both the VIP and default implementations are due to monitoring activity performed by the Cloudlab environment, and are unrelated to our techniques.}
\begin{figure*}
	\centering
	\subfloat[Medium churn rate (\S\ref{sec:popularity_churn}) with $zipf=1$. Sensing triggers learning 3 times, when it detects a significant difference in average displacement. \revision{Throughput increases by 18.9\% overall}.]{
		\includegraphics[scale=0.46]{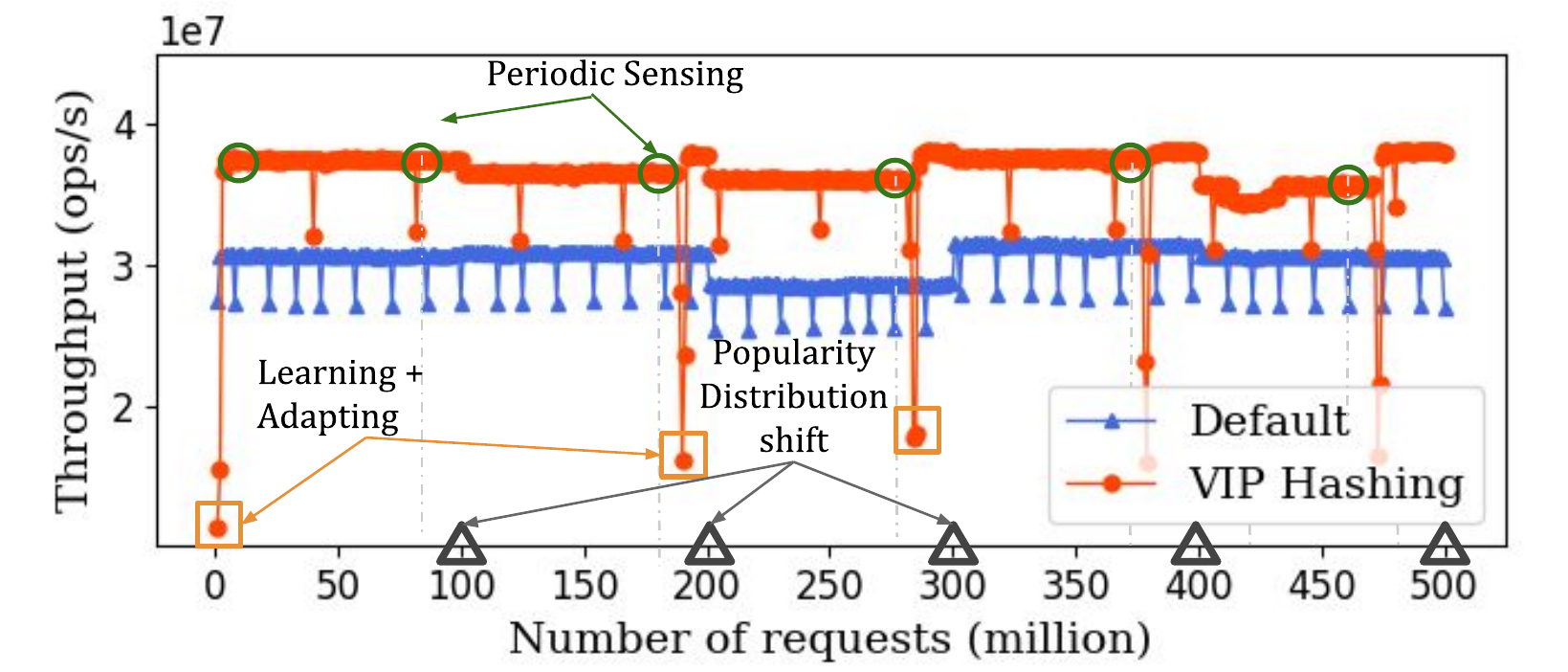}
		\label{fig:medium_churn_zipf1}
	}
	\hspace{4mm}
	\subfloat[\shepherd{Medium churn rate (\S\ref{sec:popularity_churn}) with $zipf=1.5$. Sensing triggers learning only 1 time, when it detects a significant difference in average displacement. Throughput increases by 49.0\% overall.}]{
		\includegraphics[scale=0.46]{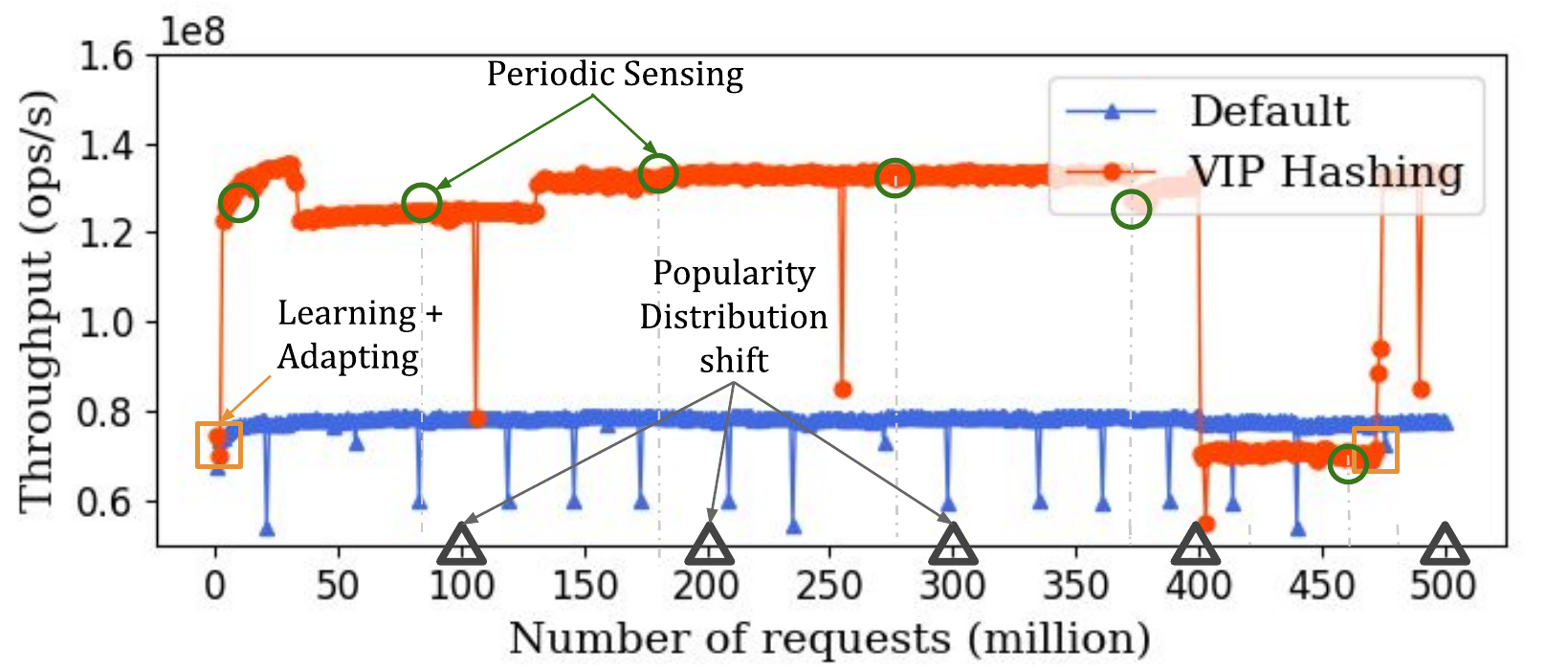}
		\label{fig:medium_churn_zipf1_5}
	}
%	\vspace{3mm}
	\caption{\textbf{\shepherd{Performance of VIP hashing under medium popularity churn at $zipf=1$ and $zipf=1.5$. Fetch requests are issued to a hash table with 1M entries at load factor 0.95 with keys in a random order initially. Popularity distribution shifts every 100M requests by 25\% (top 21 out of 1M keys at $zipf=1$, and the topmost key at $zipf=1.5$, are replaced by less popular key(s) at random). Distribution shift increases average displacement and can reduce performance (notice drop in performance of VIP hashing at 400M requests for $zipf=1.5$). Sensing triggers learning\textsuperscript{2} whenever it detects a significant increase in average displacement. Net increase in throughput for $zipf=1$ and $zipf=1.5$ is 19\% and 49\% respectively.}}}
	\label{fig:vip_results_medium_churn}
	\begin{minipage}{\textwidth}
		\vspace{6mm}
%		\hline % misuse!
		\vspace{2mm}
		\small\textsuperscript{2} \shepherd{The triggers of sense mode and learn+adapt mode have been marked using green circles and orange squares respectively. The unmarked periodic dips in throughput for both the VIP and default implementations are due to monitoring activity performed by the Cloudlab~\cite{cloudlab} environment, and are unrelated to VIP hashing.}
	\end{minipage}
\end{figure*}
\begin{figure*}
	\centering
	\subfloat[High churn rate (\S\ref{sec:popularity_churn}) with $zipf=1$. \revision{Overall, 11.8\% increase in throughput is observed.}]{
		\includegraphics[scale=0.46]{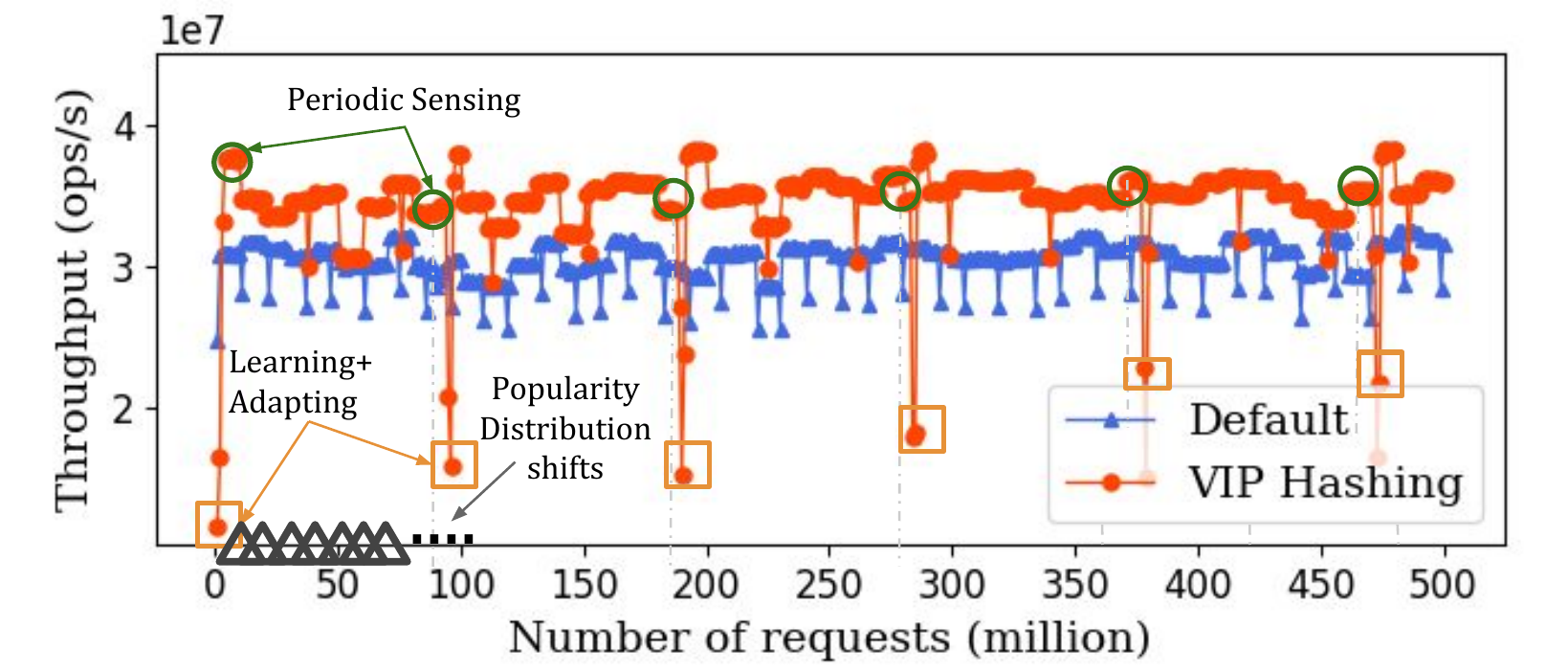}
		\label{fig:high_churn_zipf1}
	}
	\hspace{4mm}
	\subfloat[\shepherd{High churn rate (\S\ref{sec:popularity_churn}) with $zipf=1.5$. Overall, 21.9\% increase in throughput is observed.}]{
		\includegraphics[scale=0.46]{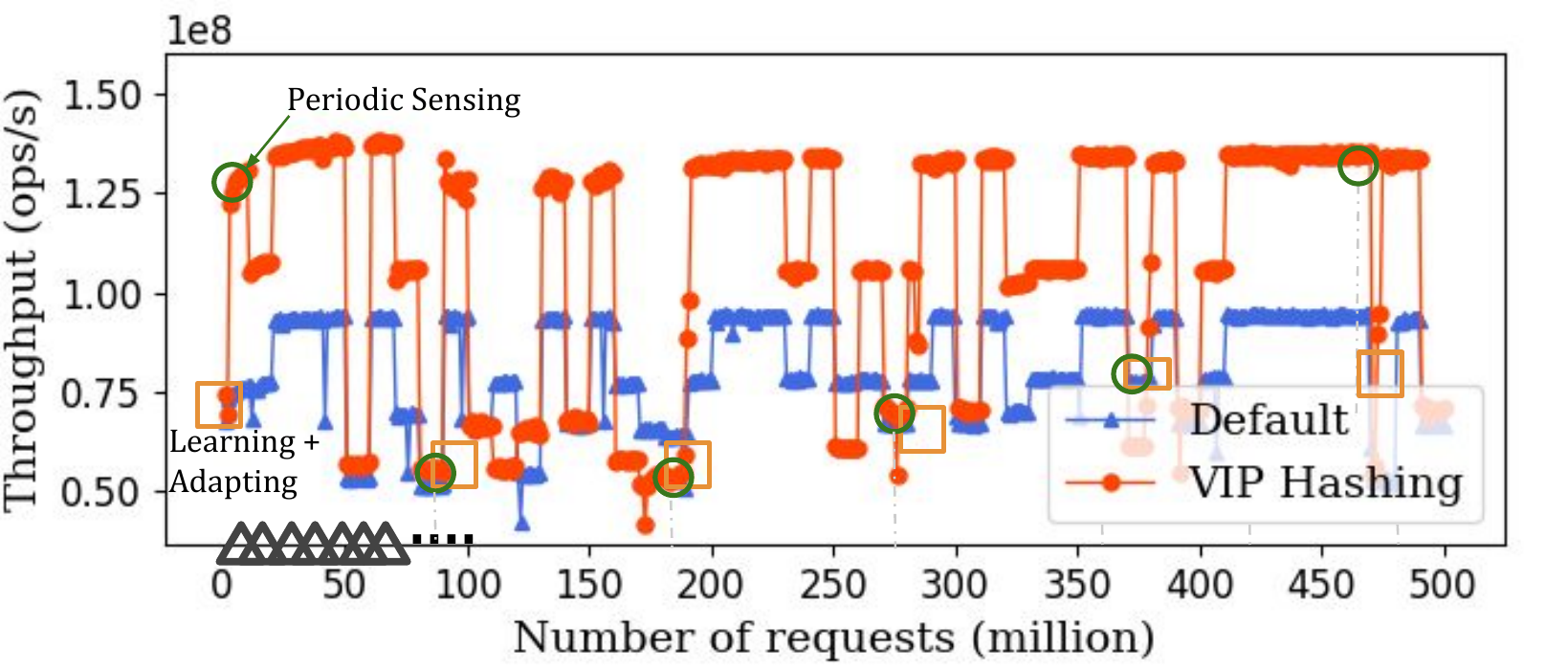}
		\label{fig:high_churn_zipf1_5}
	}
	%	\vspace{3mm}
	\caption{\textbf{\shepherd{Performance of VIP hahsing under high popularity churn at $zipf=1$ and $zipf=1.5$. Fetch requests are issued to a hash table with 1M entries at load factor 0.95 with keys in a random order initially. Popularity distribution shifts every 10M requests by 50\% (top 750 out of 1M keys at $zipf=1$, and the top 2 keys at $zipf=1.5$, are replaced by less popular keys at random). The benefit of learning dimishes as the popularity order becomes shuffled. Periodic sensing triggers learning every time, as frequent distribution shifts cause significant change in average displacement. Net increase in throughput for $zipf=1$ and $zipf=1.5$ is 12\% and 22\% respectively.}}}
	\label{fig:vip_results_high_churn}
\end{figure*}
\begin{figure*}
	\centering
	\subfloat[Steady state (\S\ref{sec:eval_steady_state}) with $zipf=1$.  \revision{An overall gain of 5.4\% is observed.}]{
		\includegraphics[scale=0.46]{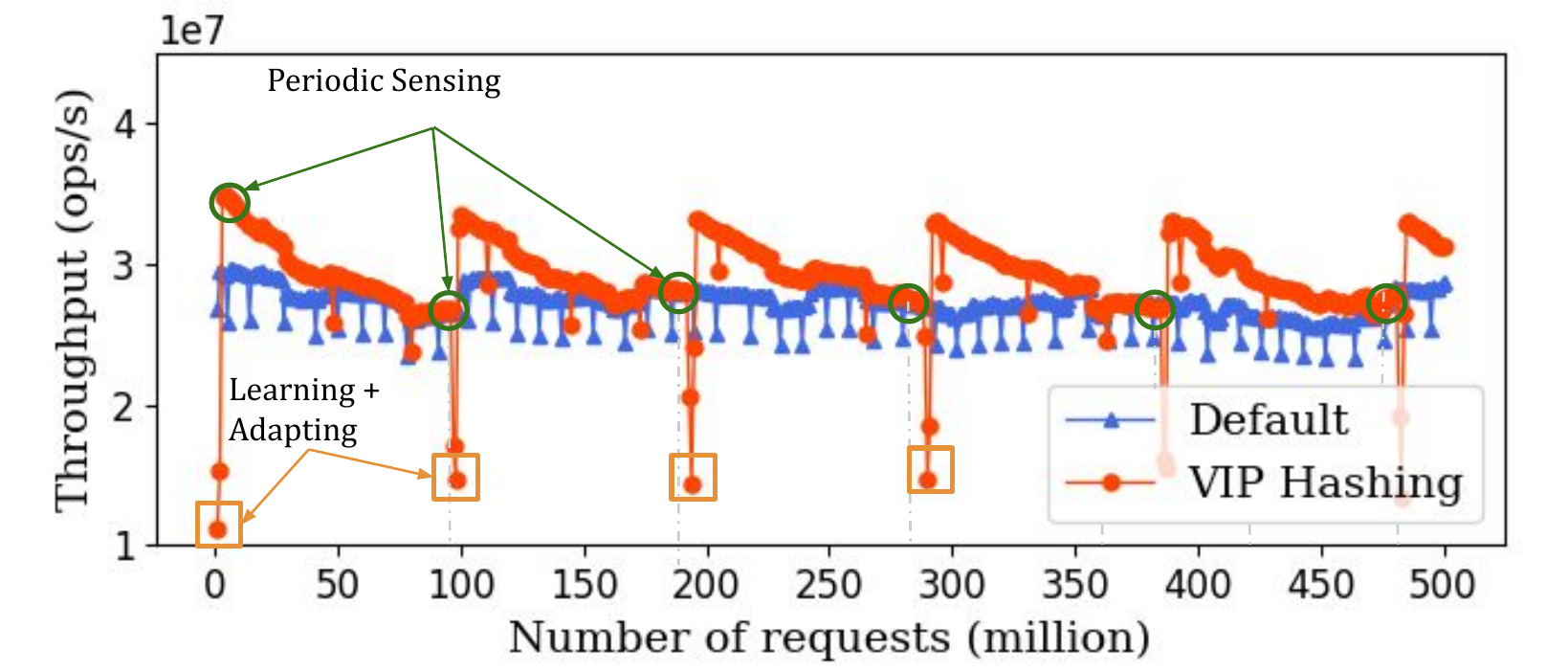}
		\label{fig:steady_state_zipf1}
	}
	\hspace{4mm}
	\subfloat[\shepherd{Steady state (\S\ref{sec:eval_steady_state}) with $zipf=1.5$.  An overall gain of 3.1\% is observed.}]{
		\includegraphics[scale=0.46]{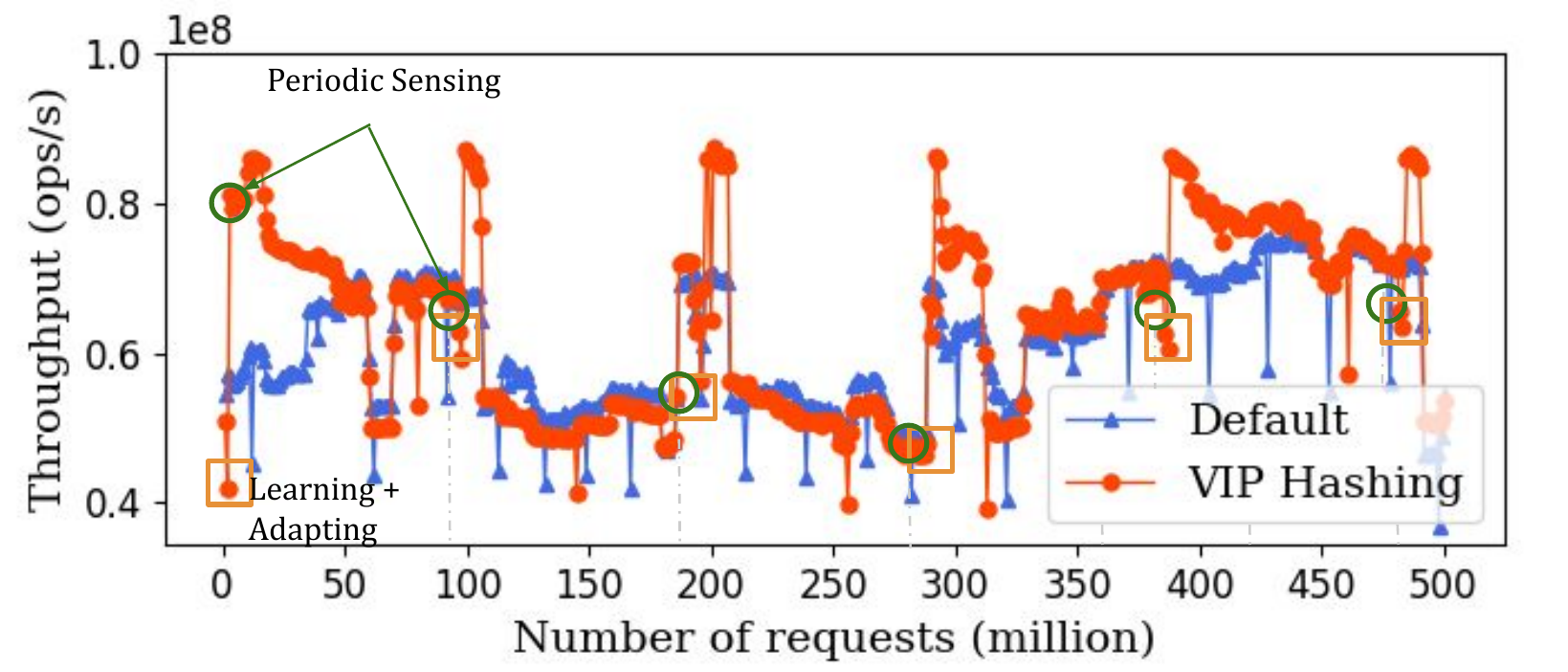}
		\label{fig:steady_state_zipf1_5}
	}
	\caption{\textbf{\shepherd{Performance of VIP hashing in a steady state consisting of 98\% fetch requests, 1\% insert requests, and 1\% delete requests at $zipf=1$ and $zipf=1.5$. 500M fetch requests are issued to a hash table with 1M entries at load factor 0.95 with keys in a random order initially. With new keys being inserted at the front of the buckets and existing keys being deleted, the hash table arrangement steadily becomes worse. Periodic sensing triggers learning every time which bounces back the performance. The net throughput gain for $zipf=1$ and $zipf=1.5$ is 5\% and 3\% respectively.}}}
	\label{fig:vip_results_steady_state}
\end{figure*}

\begin{figure*}
	\centering
	\subfloat[Ready mostly workload (\S\ref{sec:eval_read_mostly}) with $zipf=1$. \revision{Overall, we observe a gain of 1\% in throughput}.]{
		\includegraphics[scale=0.46]{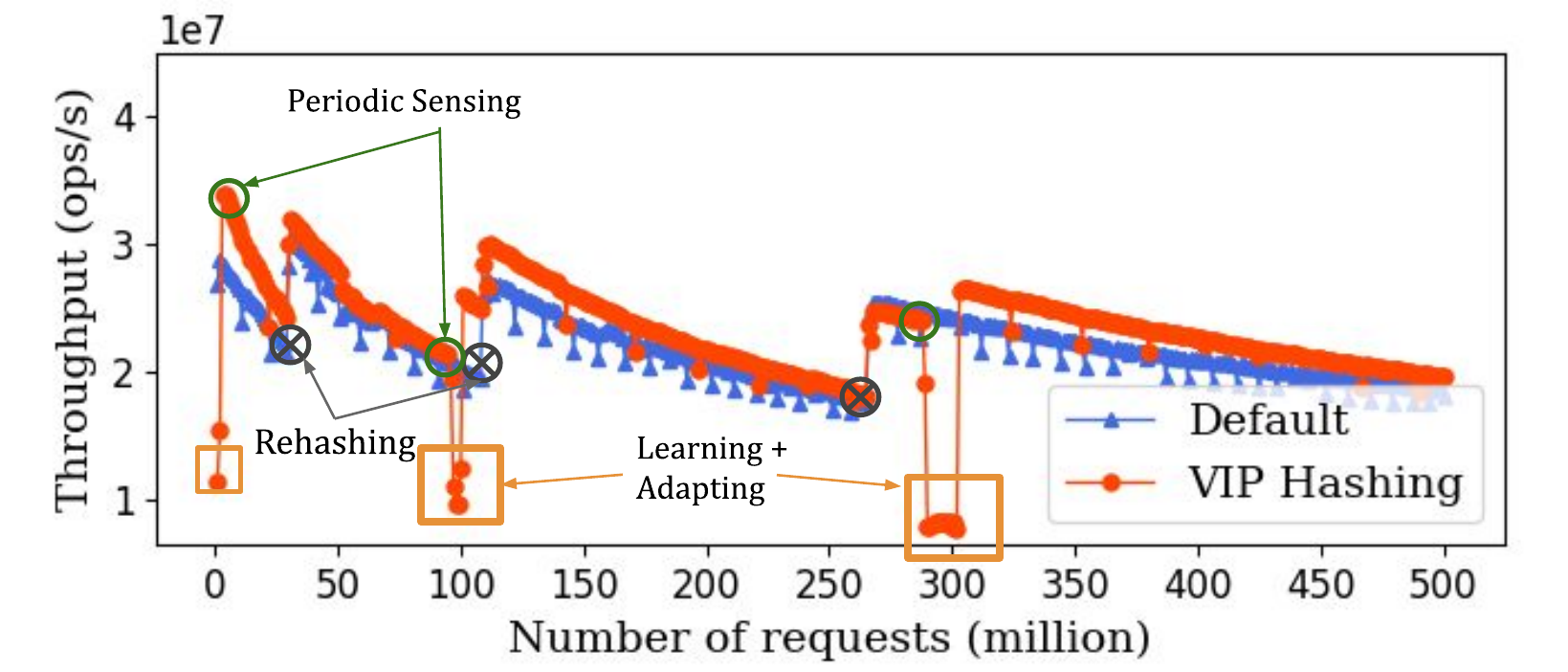}
		\label{fig:read_mostly_zipf1}
	}
	\hspace{4mm}
	\subfloat[\shepherd{Ready mostly workload (\S\ref{sec:eval_read_mostly}) with $zipf=1.5$. Overall, we observe a gain of 10.6\% in throughput}.]{
		\includegraphics[scale=0.46]{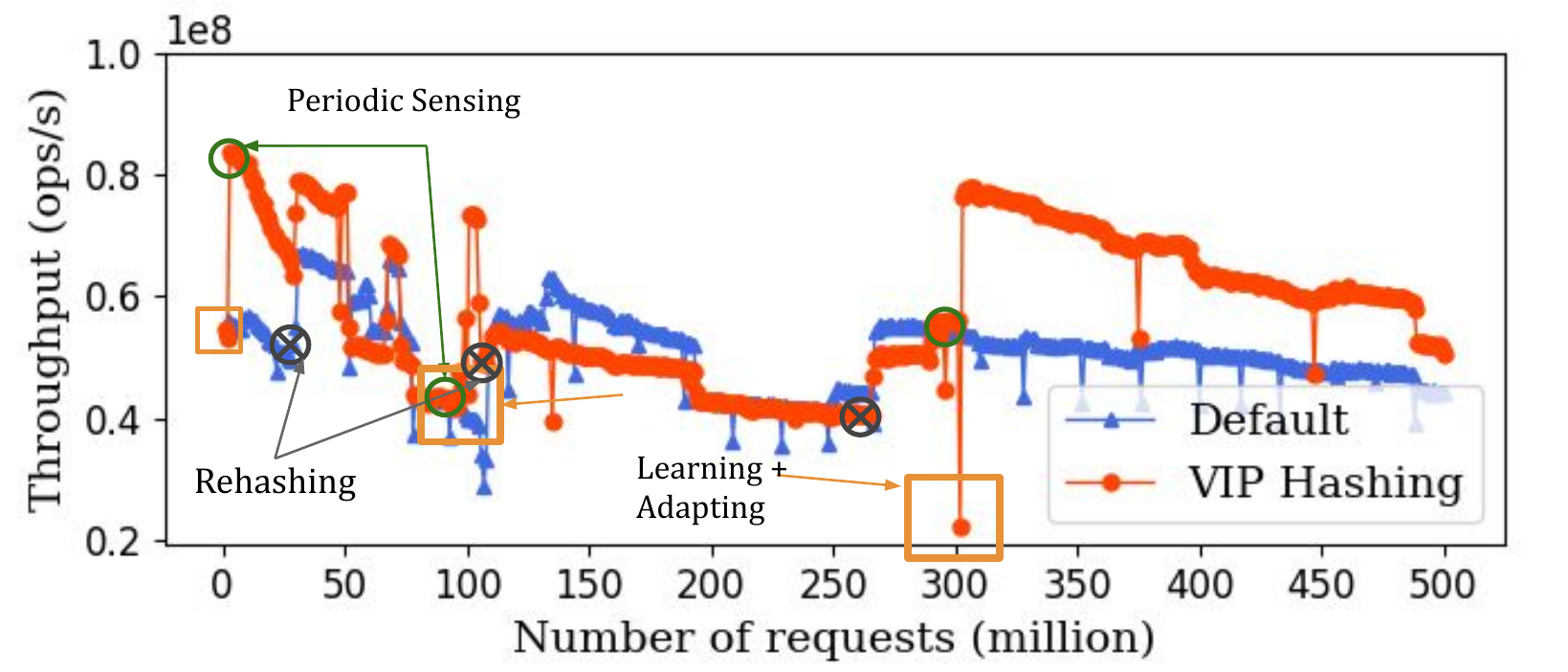}
		\label{fig:read_mostly_zipf1_5}
	}
	\caption{\textbf{\shepherd{Performance of VIP hashing under a ready mostly workload with 98\% fetch requests and 2\% insert requests at $zipf=1$ and $zipf=1.5$. 500M fetch requests are issued to a hash table with 1M entries at load factor 0.95 with keys in a random order initially. For this workload, rehashing is triggered when the load factor reaches 1.5, which happens every $75\cdot htsize$ requests. Whenever rehashing occurs, we double the periodicity of sensing (N\textsubscript{S}) and the duration of learning (N\textsubscript{L}), i.e., learning is triggered less frequently for longer duration each time. The net gain in throughput for $zipf=1$ and $zipf=1.5$ is 1\% and 11\% respectively.}}}
	\label{fig:vip_results_read_mostly}
\end{figure*}

\subsubsection{Popularity Churn}\label{sec:popularity_churn}

In this workload, we study how VIP hashing adapts to changing popularity distribution over time. We simulate two rates of shift \textendash\ medium and high. For the case of medium churn, the popularity distribution shifts by 25\% every 100M fetch operations (about 3s at $zipf=1$). Fig.~\ref{fig:vip_results_medium_churn} shows the behavior of VIP hashing under medium churn. Note that the sense mode triggers learning only when necessary. For instance, learning was triggered 3 out of the 4 times at $zipf=1$ only when there was a substantial change in average displacement due to shift in popularity (accompanied by a decrease in performance). \shepherd{The gain in throughput for $zipf=1$ and $zipf=1.5$ is 19\% and 49\% respectively.}

For the case of high churn (Fig.~\ref{fig:vip_results_high_churn}), the popularity distribution shifts by 50\% every 10M fetch operations ($<$ 1s), i.e., popularity shift occurs 50 times during the experiment. Every run of the sense mode detects a change in distribution and learning is triggered every time for both levels of skew. \shepherd{We obtain a net increase of 12\% and 22\% in throughput for  $zipf=1$ and $zipf=1.5$ respectively}. Thus, VIP hashing is able to sense changes in the distribution, and re-learn on the fly.

%\subsubsection{High Popularity Churn}\label{sec:eval_high_churn_rate}
%
%We simulate high rate of popularity shift in this workload \textendash\ the top 750 keys (cumulative probability 50\%) are randomly replaced by less popular keys every 10M requests ($<1$s). A total of 500M fetch requests are issued, and the popularity shifts 50 times during the experiment. Given that the shift in distribution is substantial, on every run of the sense mode a change in popularity is detected and learning is triggered. We obtain a net increase of \revision{12\%} in the throughput.

\subsubsection{Steady State}\label{sec:eval_steady_state}

Next, we test a workload with 98\% fetch requests, 1\% insert requests, and 1\% delete requests (Fig.~\ref{fig:vip_results_steady_state}). The cardinality of the hash table doesn't change substantially during the experiment, as the number of insert and delete operations are approximately balanced. The keys are inserted (deleted) in random positions of the popularity order. We observe that as new keys (which are less popular with high probability) are inserted at the front of the chains, the hash table arrangement steadily becomes worse and the performance of VIP hashing approaches the default for $zipf=1$. \shepherd{At $zipf=1.5$, the trend is similar, but is less stable as a small number of topmost keys carry most of the popularity weight}. A change in average displacement is sensed every time and learning is triggered, which bounces back the performance of VIP hashing. \shepherd{We obtain a net gain in throughput of 5.4\% and 3\% for $zipf=1$ and $zipf=1.5$ respectively}.

\subsubsection{Read Mostly}\label{sec:eval_read_mostly}

In this workload, we issue 98\% fetch requests and 2\% insert requests. New keys are inserted in arbitrary positions in the popularity order. Similar to \S\ref{sec:eval_steady_state}, we observe that the performance steadily becomes worse as new keys are inserted at the front of the bucket chains for $zipf=1$ (Fig.~\ref{fig:read_mostly_zipf1}). Inserting new keys increase the load factor, which degrades the throughput of the default implementation as well (Fig.~\ref{fig:read_mostly_zipf1}). \shepherd{The rate of degradation at $zipf=1.5$ as the popularity weight lies with a smaller portion of topmostly keys.} Rehashing is triggered when the load factor exceeds \revision{$1.5$} \revision{(happens every $75\cdot htsize$ requests)}, which bounces back the performance for both the default and VIP hashing implementations for both levels of skew. The periodicity at which sensing is triggered \revision{(every $90\cdot htsize$ requests)} increases every time rehashing is performed, as we update the parameters N\textsubscript{S} and N\textsubscript{L} according to the size of the hash table ($htsize$). Given that the change in the distribution is substantial, every run of the sense mode detects a change in popularity and triggers learning. \shepherd{The net gain in throughput for $zipf=1$ and $zipf=1.5$ is 1\% and 11\% respectively.}

% However, since the configuration of the hash table is changing at a faster rate (rehashing is performed every $50\cdot htsize$ requests) than we run the learn+adapt mode (every $60\cdot htsize$ requests), the gain from learning is unable to offset the overhead. Overall, we observe a net loss of 1.7\% in throughput, which is under the budget of 5\% that we set out with (\S\ref{sec:vip_hashing_parameters})

%\vspace{8pt}
%Overall, these experiments show the robustness of the VIP hashing method given the choice of the parameters. We obtain a gain ranging from 6-19\% depending on the workload, and the loss in throughput observed is 2.5\% in the worst case because of our choice of parameters.

	\section{Related Work}\label{sec:related_work}
	Hash tables are well studied data structures in literature. Two major categories of hash tables are chained hashing~\cite{hashtable} where collisions are resolved by chaining (\S\ref{sec:background_ht}), and open addressing~\cite{openaddressing} where collisions are resolved by searching for alternate positions in an array. Richter et al.~\cite{richter} study different hash table implementations spanning both the categories, hash functions, workload patterns, etc. while highlighting the variability in the performance of hash tables based on a host of factors. Similar to our work, they consider the problem of hashing 8-byte integer keys and values.

Multiple open source hash tables~\cite{inteltbb,bytell,flathash} use both categories of implementations. For instance, Google's flat hash table~\cite{flathash} takes an open addressing approach, while the bytell (byte linked list) hash table~\cite{bytell} uses chaining to resolve collisions. When it comes to data systems, DBMS such as SQLite3~\cite{sqlite} and PostgreSQL~\cite{postgres}, as well as key-value stores such as Redis~\cite{redisBlog} and Memcached~\cite{memcached} use data structures that involve chaining of entries. Thus, we find
that chained hash tables are a popular choice commonly
used in practice.

%But for most datsa systems with variable sized elements, chained hashing is used (cite sqlite, postgres, memcached, redis).

Skew in popularity is a well studied phenomenon. Multiple studies involving production workloads have found fetch requests to follow a power-law behavior~\cite{fbstudy,zipf2}, \revision{which is often captured using the zipfian distribution~\cite{spyros,anna,ycsb}. For instance, the request distribution in the core workloads of YCSB~\cite{ycsbCoreWorkloads} is zipfian by default. Alongside skew in popularity, previous work~\cite{fbstudy} also discusses effects such as churn in popular keys in real world workloads. This is a key feature captured by Wiscer (\S\ref{sec:wiscer}), which is not present in any of the existing workload generators to the best of our knowledge.}

%\revision{, which has been modeled using zipfian distribution~\cite{zipf} in multiple studies~\cite{spyros,anna}. The core workloads of YCSB~\cite{ycsb}, a popular workload generator, follow zipfian request distribution. Previous work~\cite{fbstudy} discusses effects such as churn in popularity of keys in real world workloads, but none of the existing workload generators capture this behavior to the best of our knowledge. Simulating shifts in popularity distribution is a key feature of Wiscer (\S\ref{sec:wiscer})}.

Broadly speaking, caching algorithms such as LRU-k~\cite{lruk}, MRU~\cite{mru}, etc. that track the recency of access are attempting to capture the current popularity distribution. Key-value stores designed for disk-based settings, such as Anna~\cite{anna} and Faster~\cite{faster} incorporate techniques to leverage the skew in popularity by moving hot data to memory. Recent work by Herodotou et al.~\cite{herodotus} uses machine learning to automatically move data between different storage tiers in clusters. A recurring trend to note here is that the complexity of these existing schemes vary depending on the ``budget'' allowed by the setting, ranging from relatively simple LRU approach is used even in processor caches, to a more complex approach involving machine learning in large-scale clusters.

\revision{To this end, the budget available for learning in-the-loop with hash tables is extremely limited, as we see in our work (Fig.~\ref{fig:adding_counter}). In the seminal paper on learned indexes by Kraska et al.~\cite{kraska}, the authors propose learning a hash function from the keys in the hash table such that collisions can be avoided altogether. However, recent work on learned hash functions~\cite{learned_hashmap_recent} shows that this approach hits the wall due to two main reasons \textendash\ cache sensitivity, and model complexity. While larger models are necessary to accurately capture arbitrary key distributions, the computation times become prohibitively high (50\textit{x} higher~\cite{learned_hashmap_recent}) due to increased cache misses from accessing the model parameters. The high cache sensitivity and low latency requirements of hash tables preclude the use of costly ML techniques for learning.}

A noteworthy aspect of the VIP hashing method is that learning is performed online, i.e., the hash table does not pause operation at any time. In contrast, recent work~\cite{learned_hashmap_recent,entropylearnedhashing} involves learning from the data offline before populating the hash table. Adapting to changing key distributions remains a challenge with these approaches, with the fallback being reverting to the default hash table implementation~\cite{entropylearnedhashing} or relearning~\cite{learned_hashmap_recent,kraska}, both of which require costly rehashing that pauses execution.

%A drawback of this proposition is the intolerance to insert operations \textendash\ as new keys are inserted, one might need to re-learn the hash function while remaining consistent with the older version, or perform a rehash. Moreover, using techniques like neural networks can be very costly given the low latency setting imposed by hash tables. To our knowledge, our work is the first attempt to leverage the skew in popularity to improve the performance of hash tables, by learning in-the-loop in a lightweight manner.

%\revision{Recent work highlights challenges of using ML models as hash functions, and suffer from high latency and cache misses, while training times can be in the order of minutes.}

%\revision{When it comes to workload generation tools, YCSB provides options to generate skewed data with zipfian being the default. However, a key feature that Wiscer implements - shifting popularity distribution over time - is not captured by any other workload generation tool to the best of our knowledge.}
	
	\section{Conclusions \& Future Work}\label{sec:conclusion}
	Hashing is a low-latency operation that runs a tight loop of operations, and is sensitive to the effects of caching. Increasing the memory and computation footprint even by a small proportion can have a significant impact on performance as we see in \S\ref{sec:learning_is_costly}. Given these constraints, learning in-the-loop precludes the use of costly techniques and makes it necessary to use lightweight schemes while controlling the overhead as much as possible.

Overall, VIP hashing is comprised of four mechanisms \textendash\ \textit{learning}, \textit{adapting}, \textit{sensing}, and \textit{dynamically switching-on/off} learning. These mechanisms (\S\ref{sec:vip_hashing}), along with our choice of parameters (\S\ref{sec:vip_hashing_parameters}) keep the overhead of learning in check compared to the gains. \revision{We evaluate VIP hashing using an extensive set of workloads (Fig.~\ref{fig:eval_hashjoin}-\ref{fig:vip_results_read_mostly}) that demonstrate the ability to learn on the fly in the presence of insert and delete operations, and shifting distributions.} \shepherd{Our experiments involving PK-FK hash joins show that VIP hashing reduces the end-to-end execution time by 22\%, while the gain in performance for point queries ranges from 3\%-77\% under medium skew. While the performance gain depends on the a host of factors (level of skew, proportion of insert and delete, etc.), the distinguishing property of VIP hashing is the ability to learn in a non-blocking, online fashion.}

Broadly speaking, our work highlights the challenges of learning with cache sensitive, low latency data structures. \revision{While the major source of performance gain for VIP hashing has been from improvement in cache locality, the sensitivity of hash tables to effects of caching make learning very challenging (\S\ref{sec:learning_is_costly},~\cite{learned_hashmap_recent})}. Possible future work could involve studying other low latency data structures such as bloom filters~\cite{bloomfilter}, to see how cache locality can be improved by adapting to the data. \revision{Learning tasks involving such cache sensitive data structures will necessitate controlling the overhead, perhaps by using our approach of budgeted learning and non-intrusive sensing.}

%our approach of budgeting the cost of learning along with using lightweight non-intrusive sensing mechanisms to control the overhead.

% The reason VIP hashing performs better is due to improvement in cache locality.

%You can implement it in real systems. Our approach of lightweight sensing and running only when required can be applied to different parts of the database system. Maybe look at other low latency data structure such as bloom filters and bitmaps. Same approach of sensing and running only when required will need to be applied.
	
	\section{Acknowledgments}
	This research was supported in part by a grant from the Microsoft Jim Gray Systems Lab, by the National Science Foundation under grant OAC-1835446, and by CRISP, one of six centers in JUMP, a Semiconductor Research Corporation (SRC) program, sponsored by MARCO and DARPA.
	
	\bibliographystyle{abbrv}
	\bibliography{VIPHashing}

\begin{thebibliography}{10}

\bibitem{redisBlog}
{A little internal on Redis hash table implementation}.
\newblock \url{https://bit.ly/3pfVvTm}.

\bibitem{bytell}
{Bytell hash map}.
\newblock \url{https://bit.ly/3fB8NX6}.

\bibitem{redis}
{Data types in Redis}.
\newblock \url{https://redis.io/topics/data-types}.

\bibitem{mysql}
{Hash join in MySQL 8}.
\newblock \url{https://mysqlserverteam.com/hash-join-in-mysql-8/}.

\bibitem{hashtable}
{Hash table}.
\newblock \url{https://en.wikipedia.org/wiki/Hash_table}.

\bibitem{hellingers}
{Hellinger's distance}.
\newblock \url{https://en.wikipedia.org/wiki/Hellinger_distance}.

\bibitem{postgres}
{Indexes in PostgreSQL}.
\newblock \url{https://bit.ly/3c7L52A}.

\bibitem{intelintrinsics}
{Intel Intrinsics}.
\newblock \url{https://intel.ly/3nxA416}.

\bibitem{pmu}
{Intel performance monitoring events}.
\newblock \url{https://perfmon-events.intel.com/}.

\bibitem{inteltbb}
{Intel TBB hash map}.
\newblock \url{https://intel.ly/3uDtNAQ}.

\bibitem{cpu}
{Intel Xeon Silver 4114 processor}.
\newblock \url{https://intel.ly/3fDidSb}.

\bibitem{kldivergence}
{Kullback-Leibler divergence}.
\newblock
  \url{https://en.wikipedia.org/wiki/Kullback%E2%80%93Leibler_divergence}.

\bibitem{clt}
{Lindeberg-Levy CLT}.
\newblock \url{https://bit.ly/34A19WJ}.

\bibitem{mariadbhash}
{MariaDB Storage Index Types}.
\newblock \url{https://mariadb.com/kb/en/storage-engine-index-types/}.

\bibitem{murmurhash}
{MurmurHash3}.
\newblock \url{https://github.com/aappleby/smhasher/wiki/MurmurHash3}.

\bibitem{openaddressing}
{Open addressing}.
\newblock \url{https://en.wikipedia.org/wiki/Open_addressing}.

\bibitem{powerlaw}
{Power Law}.
\newblock \url{https://en.wikipedia.org/wiki/Power_law}.

\bibitem{sqlite}
{SQLite hash table implementation}.
\newblock \url{https://sqlite.org/src/file/src/hash.c}.

\bibitem{flathash}
{Swiss Tables and absl::Hash}.
\newblock \url{https://abseil.io/blog/20180927-swisstables}.

\bibitem{tpch}
{TPC-H Benchmark (Version 3)}.
\newblock \url{http://www.tpc.org/tpch/}.

\bibitem{memcached}
{Understanding the Memcached source code}.
\newblock \url{https://holmeshe.me/understanding-memcached-source-code-V/}.

\bibitem{wiscer_bib}
{Wiscer}.
\newblock \url{https://github.com/aarati-K/wiscer}.

\bibitem{ycsbCoreWorkloads}
{YCSB Core Workloads}.
\newblock \url{https://github.com/brianfrankcooper/YCSB/wiki/Core-Workloads}.

\bibitem{ztest}
{Z-test}.
\newblock \url{https://en.wikipedia.org/wiki/Z-test}.

\bibitem{zipf}
{Zipf's law}.
\newblock \url{https://bit.ly/3yTN0BO}.

\bibitem{fbstudy}
B.~Atikoglu, Y.~Xu, E.~Frachtenberg, S.~Jiang, and M.~Paleczny.
\newblock {Workload Analysis of a Large-Scale Key-Value Store}.
\newblock {\em Sigmetrics Performance Evaluation Review - SIGMETRICS}, 2012.

\bibitem{alonso}
C.~Balkesen, J.~Teubner, G.~Alonso, and M.~T. Ozsu.
\newblock {Main-memory hash joins on multi-core CPUs: Tuning to the underlying
  hardware}.
\newblock In {\em 2013 IEEE 29th International Conference on Data Engineering}.

\bibitem{spyros}
S.~Blanas, Y.~Li, and J.~M. Patel.
\newblock {Design and Evaluation of Main Memory Hash Join Algorithms for
  Multi-Core CPUs}.
\newblock In {\em Proceedings of the 2011 ACM SIGMOD International Conference
  on Management of Data}. Association for Computing Machinery.

\bibitem{bloomfilter}
B.~Bloom.
\newblock {Space/time trade-offs in hash coding with allowable errors}.
\newblock {\em Commun. ACM}, 1970.

\bibitem{zipf2}
L.~Breslau, P.~Cao, L.~Fan, G.~Phillips, and S.~Shenker.
\newblock {Web caching and Zipf-like distributions: evidence and implications}.
\newblock In {\em IEEE INFOCOM '99}.

\bibitem{statisticalcomplexity}
C.~L. Canonne.
\newblock {A short note on learning discrete distributions}.
\newblock {\em arXiv: Statistics Theory}, 2020.

\bibitem{faster}
B.~Chandramouli, G.~Prasaad, D.~Kossmann, J.~Levandoski, J.~Hunter, and
  M.~Barnett.
\newblock {FASTER: A Concurrent Key-Value Store with In-Place Updates}.
\newblock In {\em 2018 ACM SIGMOD International Conference on Management of
  Data (SIGMOD '18)}.

\bibitem{mru}
H.~Chou and D.~DeWitt.
\newblock {An Evaluation of Buffer Management Strategies for Relational
  Database Systems}.
\newblock {\em Algorithmica}, 2005.

\bibitem{ycsb}
B.~F. Cooper, A.~Silberstein, E.~Tam, R.~Ramakrishnan, and R.~Sears.
\newblock {Benchmarking Cloud Serving Systems with YCSB}.
\newblock In {\em Proceedings of the 1st ACM Symposium on Cloud Computing},
  SoCC 2010.

\bibitem{cloudlab}
D.~Duplyakin, R.~Ricci, A.~Maricq, G.~Wong, J.~Duerig, E.~Eide, L.~Stoller,
  M.~Hibler, D.~Johnson, K.~Webb, A.~Akella, K.~Wang, G.~Ricart, L.~Landweber,
  C.~Elliott, M.~Zink, E.~Cecchet, S.~Kar, and P.~Mishra.
\newblock {The Design and Operation of CloudLab}.
\newblock In {\em Proceedings of the {USENIX} Annual Technical Conference
  (ATC)}, 2019.

\bibitem{entropylearnedhashing}
B.~Hentschel, U.~Sirin, and S.~Idreos.
\newblock Entropy-learned hashing: Constant time hashing with controllable
  uniformity.
\newblock In {\em Proceedings of the 2022 International Conference on
  Management of Data}, SIGMOD '22, 2022.

\bibitem{herodotus}
H.~Herodotou and E.~Kakoulli.
\newblock {Automating distributed tiered storage management in cluster
  computing}.
\newblock {\em Proc. of the VLDB Endowment}, 2019.

\bibitem{hjpaper}
M.~Kitsuregawa, H.~Tanaka, and T.~Moto-Oka.
\newblock {Application of hash to data base machine and its architecture}.
\newblock {\em New Generation Computing}, 2009.

\bibitem{kraska}
T.~Kraska, A.~Beutel, E.~H. Chi, J.~Dean, and N.~Polyzotis.
\newblock {The Case for Learned Index Structures}.
\newblock {\em CoRR}, abs/1712.01208, 2017.

\bibitem{lruk}
E.~O'neil, P.~O'Neil, G.~Weikum, and E.~Zurich.
\newblock {The LRU--K Page Replacement Algorithm For Database Disk Buffering}.
\newblock {\em SIGMOD Record (ACM Special Interest Group on Management of
  Data)}, 1996.

\bibitem{duckdb}
M.~Raasveldt and H.~M\"{u}hleisen.
\newblock {DuckDB: An Embeddable Analytical Database}.
\newblock In {\em Proceedings of the 2019 International Conference on
  Management of Data}, SIGMOD '19.

\bibitem{richter}
S.~Richter, V.~Alvarez, and J.~Dittrich.
\newblock {A Seven-Dimensional Analysis of Hashing Methods and Its Implications
  on Query Processing}.
\newblock {\em Proceedings of the VLDB Endowment}, 2015.

\bibitem{learned_hashmap_recent}
I.~Sabek, K.~Vaidya, D.~Horn, A.~Kipf, and T.~Kraska.
\newblock {When Are Learned Models Better Than Hash Functions?}
\newblock {\em CoRR}, abs/2107.01464, 2021.

\bibitem{anna}
C.~Wu, V.~Sreekanti, and J.~Hellerstein.
\newblock {Autoscaling tiered cloud storage in Anna}.
\newblock {\em Proceedings of the VLDB Endowment}, 2019.

\end{thebibliography}
	\begin{appendix}
		\section{Proof of Theorem 1}\label{sec:appendix_a}
		Theorem 1 (\S\ref{sec:vip_hashing_learning_adapting}) states that given keys $K_1,\ K_2,\ $
$..., K_n$ in a bucket with probability $p_1>p_2>..>p_n$, such that the keys are in a random order initially. Then by applying Algorithm~\ref{alg:learning}, the keys will converge to the sorted order of popularity as the number of fetch requests $N\rightarrow \infty$. \revision{We first make the following observation:}
\revision{\begin{lemma}
	Given two keys $K_1$ and $K_2$ with popularity $p$ and $(1-p)$ respectively. Let $p>0.5$. Given $N$ successful fetch requests are made, and keys $K_1$ and $K_2$ receive $N_1$ and $N_2$ requests respectively. Then,
	\begin{displaymath}
		\lim\limits_{N\rightarrow\infty}\frac{N_1-N_2}{N} = (2\cdot p - 1) > 0
	\end{displaymath}
%	i.e., the probability that $K_2$ receives more requests than $K_1$ is zero as $N\rightarrow 0$.
\end{lemma}}
%\newcommand*{\Comb}[2]{{}^{#1}C_{#2}}%

%\begin{proof}
%	\begin{align*}
%		P(N_2>N_1)&\leq P\Big(N_2\geq\frac{N}{2}\Big)&&\\
%%		&\leq P\Big(N_2\geq\frac{N}{2}\Big)&&\\
%		&=\sum\limits_{i=\frac{N}{2}}^{N}\Comb{N}{i}\cdot(1-p)^i\cdot p^{(N-i)}&&\\
%		&=p^N\sum\limits_{i=\frac{N}{2}}^{N}\Comb{N}{i}\bigg(\frac{1-p}{p}\bigg)^i&&\\
%		&\leq p^N\cdot\bigg(\frac{1-p}{p}\bigg)^{\frac{N}{2}}\cdot 2^{N}&&\hspace{-8pt}(\text{since}\ \sfrac{1-p}{p}<1)\\
%		&=(4\cdot p\cdot (1-p))^{\frac{N}{2}}\rightarrow 0&&\hspace{-8pt}(p\cdot(1-p) < 0.25)
%	\end{align*}
%\end{proof}
\begin{table*}[]
	\begin{tabular}{|c|c|c|l|}
		\hline
		\begin{tabular}[c]{@{}c@{}}\textbf{Dimension}\\ \textit{(Primary Key)}\end{tabular}             & \begin{tabular}[c]{@{}c@{}}\textbf{Fact}\\ \textit{(Foreign Key)}\end{tabular}                & \begin{tabular}[c]{@{}c@{}}\textbf{Skew added}\\ \textbf{to FK?}\end{tabular} & \multicolumn{1}{c|}{\textbf{Comments}} \\ \hline
		
		\begin{tabular}[c]{@{}c@{}}Part\\ \textit{(p\_partkey)}\end{tabular}                   & \begin{tabular}[c]{@{}c@{}}PartSupp\\ \textit{(ps\_partkey)}\end{tabular}        & No                                                          & \begin{tabular}[c]{@{}l@{}}Each part has a fixed number of suppliers (4 suppliers\\per part). Thus, \textit{ps\_partkey} cannot be skewed.\end{tabular} \\ \hline
		
		\begin{tabular}[c]{@{}c@{}}Supplier\\ (s\_suppkey)\end{tabular}               & \begin{tabular}[c]{@{}c@{}}PartSupp\\ (ps\_suppkey)\end{tabular}        & Yes                                                         & \begin{tabular}[c]{@{}l@{}}\textit{ps\_suppkey} is zipfian distributed, i.e., a supplier is chosen\\from a zipfian distribution over \textit{s\_suppkey}. We ensure\\that each part has 4 distinct suppliers.\end{tabular} \\ \hline
		
		\begin{tabular}[c]{@{}c@{}}PartSupp\\ \textit{(ps\_partkey, ps\_suppkey)}\end{tabular} & \begin{tabular}[c]{@{}c@{}}Lineitem\\ \textit{(l\_partkey, l\_suppkey)}\end{tabular} & Yes                                                         & \begin{tabular}[c]{@{}l@{}}\textit{l\_partkey} is zipfian distributed. \textit{l\_suppkey} is picked\\randomly from the available suppliers of the chosen part.\end{tabular} \\ \hline
		
		\begin{tabular}[c]{@{}c@{}}Customer\\ \textit{(c\_custkey)}\end{tabular}               & \begin{tabular}[c]{@{}c@{}}Orders\\ \textit{(o\_custkey)}\end{tabular}               & Yes                                                         & \begin{tabular}[c]{@{}l@{}}\textit{o\_custkey} is zipfian distributed, i.e., \textit{o\_custkey} is\\drawn from a zipfian distribution over all \textit{c\_custkey}.\end{tabular} \\ \hline
		
		\begin{tabular}[c]{@{}c@{}}Orders\\ \textit{(o\_orderkey)}\end{tabular}                & \begin{tabular}[c]{@{}c@{}}Lineitem\\ \textit{(l\_orderkey)}\end{tabular}            & No                                                          & \begin{tabular}[c]{@{}l@{}}Each order \textit{(o\_orderkey)} can have limited number of\\lineitems (1 to 7). Thus, \textit{l\_orderkey} cannot be skewed.\end{tabular} \\ \hline
		
		\begin{tabular}[c]{@{}c@{}}Nation\\ \textit{(n\_nationkey)}\end{tabular}                & \begin{tabular}[c]{@{}c@{}}Supplier\\ \textit{(s\_nationkey)}\end{tabular}            & Yes                                                          & \begin{tabular}[c]{@{}l@{}}\textit{s\_nationkey} is zipfian distributed, i.e., \textit{s\_nationkey} is\\drawn from a zipfian distribution over all \textit{n\_nationkey}.\end{tabular} \\ \hline
		
		\begin{tabular}[c]{@{}c@{}}Nation\\ \textit{(n\_nationkey)}\end{tabular}                & \begin{tabular}[c]{@{}c@{}}Customer\\ \textit{(c\_nationkey)}\end{tabular}            & Yes                                                          & \begin{tabular}[c]{@{}l@{}}\textit{c\_nationkey} is zipfian distributed, i.e., \textit{c\_nationkey} is\\drawn from a zipfian distribution over all \textit{n\_nationkey}.\end{tabular} \\ \hline
		
		\begin{tabular}[c]{@{}c@{}}Region\\ \textit{(r\_regionkey)}\end{tabular}                & \begin{tabular}[c]{@{}c@{}}Nation\\ \textit{(n\_nationkey)}\end{tabular}            & No                                                          & \begin{tabular}[c]{@{}l@{}}Each nation belongs to a fixed region (continent).\end{tabular} \\ \hline
	\end{tabular}
	\caption{\textbf{\shepherd{Introducing skew in TPC-H relations. We introduce skew in the FK attribute wherever possible under existing constraints. For 5 out of 8 cases where skew was introduced, the level of skew can be configured through the zipfian coefficient.}}}\label{tab:tpch_skew}
\end{table*}

\revision{The above lemma follows from the frequentist definition of probability}. Thus, as $N\rightarrow\infty$, we can be sure that more popular keys will receive more requests. This will hold pairwise for all the keys $K_1,\ K_2,\ $
$..., K_n$ in the bucket chain, which motivates the following claim.
\begin{lemma}
	Let $\{K_i\}$ be keys in a bucket with probability $\{p_i\},\ i\in[N]$. Let $K_1$ be the most popular key in the bucket, i.e., $p_1>p_j\ \forall j\in\{2,..,N\}$. Let the initial order of keys be random. Then, by running Algorithm~1, $K_1$ will be at the front of the chain as number of fetch requests $N\rightarrow\infty$.
\end{lemma}
\begin{proof}
	Suppose $K_1$ is at displacement $d > 1$. Let there be keys $K'_1,\ ..,\ K'_{d-1}$ in front of $K_1$. Let the keys have received requests $n_1,\ ..,\ n_{d-1}$. Let $K_1$ have received $n$ requests. From Lemma 2, we know that
	\begin{displaymath}
		\lim\limits_{N\rightarrow\infty}n > n_i,\ \forall\ i\in[(d-1)]
	\end{displaymath}
	Thus, $K_1$ would have received more requests than all the keys in front of it as $N\rightarrow\infty$. From Algorithm~1, on the last request that $K_1$ received, it should have been swapped with a key with lower number of requests ahead of it. This contradicts our assumption that $K_1$ is at position $d>1$.
\end{proof}

Thus, the most popular key in the chain will be in the front as number of requests approaches infinity. By recursively applying Lemma 3 to the remaining keys in the bucket, we can prove that the keys will be in the sorted order of popularity as $N\rightarrow\infty$.
		\section{Introducing Skew in TPC-H}\label{sec:appendix_b}
		\shepherd{Fig.~\ref{fig:tpch_schema} shows the PK-FK (primary key-foreign key) constraints in TPC-H schema. Skew can arise in PK-FK relations when a some primary keys occur more frequently than others in the fact (FK) relation, i.e., the distribution of the FK attribute is skewed. Note that primary keys are unique, and thus by definition, skew cannot arise in the PK attribute. We considered the existing constraints in TPC-H schema (Fig.~\ref{fig:tpch_schema}), and introduced skew in the FK attribute wherever possible. Table~\ref{tab:tpch_skew} details our findings \textendash\ we have introduced skew in 5 out of 8 FK attributes, and we also describe the reasons for cases where skew could not be introduced. Wherever applicable, the level of skew can be adjusted by configuring the zipfian coefficient.}

\begin{figure}
	\centering
	\vspace{-5mm}
	\includegraphics[scale=0.36]{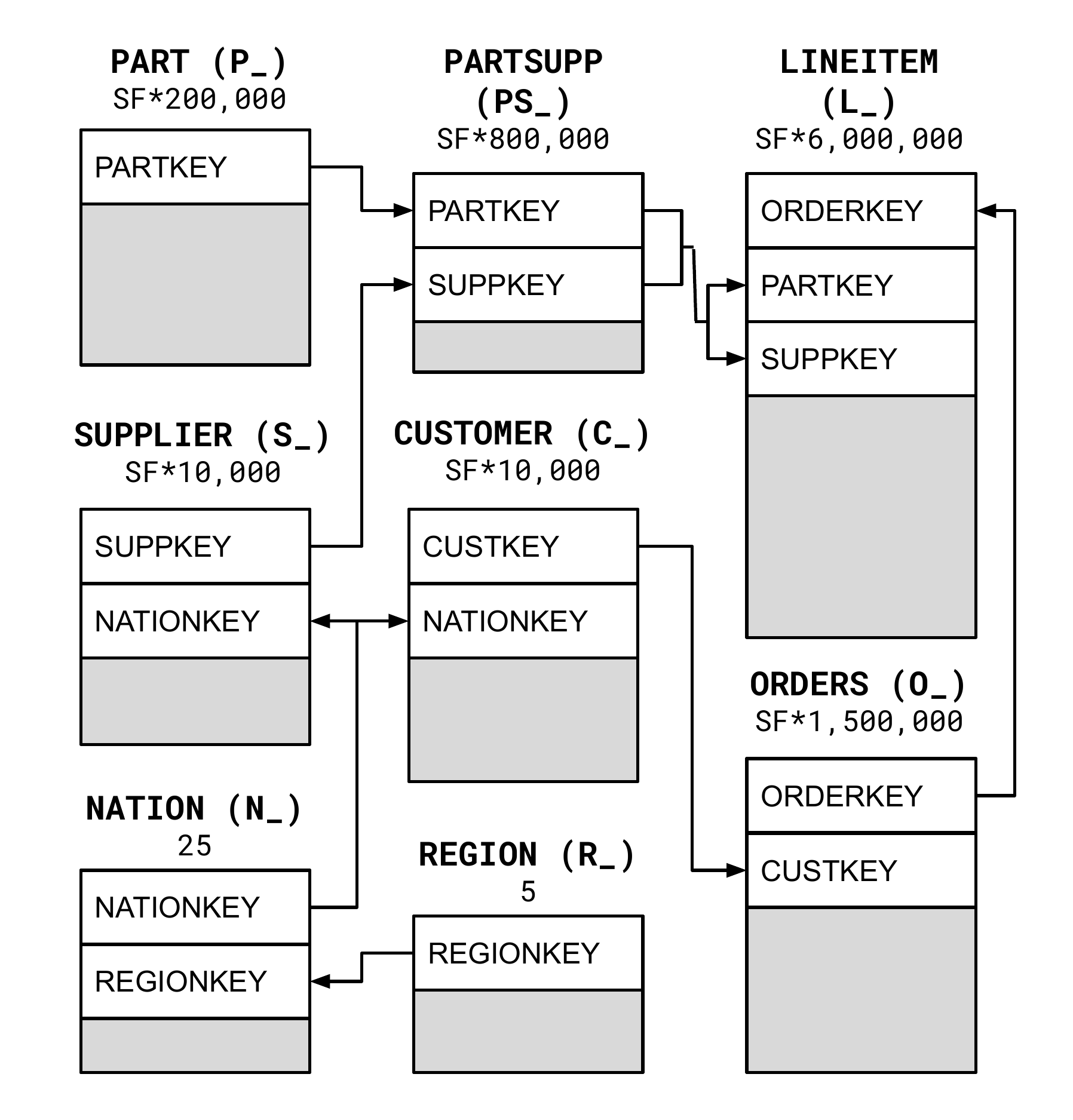}
	\caption{\shepherd{\textbf{PK-FK constraints in TPC-H schema. The cardinalities of the tables have been indicated at the top (SF denotes scale factor), and only primary key and foreign key attributes have been shown.}}}
	\label{fig:tpch_schema}
\end{figure}
	\end{appendix}
	
	\balance
\end{document}